\documentclass[pdflatex,sn-mathphys-ay]{sn-jnl}

\usepackage[utf8]{inputenc}
\usepackage[T1]{fontenc}
\usepackage[english]{babel}
\usepackage{graphicx}%
\usepackage{multirow}%
\usepackage{amsmath,amssymb,amsfonts}%
\usepackage{amsthm}%
\usepackage{mathrsfs}%
\usepackage[title]{appendix}%
\usepackage[dvipsnames]{xcolor}
\usepackage{textcomp}%
\usepackage{manyfoot}%
\usepackage{booktabs}%
\usepackage{algorithm}%
\usepackage{algorithmicx}%
\usepackage{algpseudocode}%
\usepackage{listings}%
\usepackage{cleveref}

\usepackage{mathtools}
\usepackage{pdflscape}
\usepackage{etoolbox}
\usepackage{dsfont}

\usepackage[autostyle,italian=guillemets]{csquotes}

\usepackage{rotating}

\usepackage{array,multirow}
\usepackage{makecell}
\usepackage{colortbl}

\usepackage{hyphenat}
\usepackage{enumitem}
\usepackage{siunitx}

\usepackage[mathlines, pagewise]{lineno}

\newcommand*\patchAmsMathEnvironmentForLineno[1]{%
\expandafter\let\csname old#1\expandafter\endcsname\csname #1\endcsname
\expandafter\let\csname oldend#1\expandafter\endcsname\csname end#1\endcsname
\renewenvironment{#1}%
{\linenomath\csname old#1\endcsname}%
{\csname oldend#1\endcsname\endlinenomath}%
}

\AtBeginDocument{%
\patchAmsMathEnvironmentForLineno{equation*}%
\patchAmsMathEnvironmentForLineno{align*}%
\patchAmsMathEnvironmentForLineno{align}%
\patchAmsMathEnvironmentForLineno{alignat}%
\patchAmsMathEnvironmentForLineno{gather}%
\patchAmsMathEnvironmentForLineno{multline}%
}

\newcommand{\dd}{\mathop{}\!\mathrm{d}}

\DeclareMathOperator*{\e}{e}

\DeclareMathOperator{\diag}{diag}

\usepackage{tikz}
\usepackage{verbatim}
\usetikzlibrary{arrows, positioning, quotes, shapes, overlay-beamer-styles}
\usetikzlibrary{arrows.meta}
\usetikzlibrary{positioning}
\usetikzlibrary{automata, arrows, positioning, calc}

\newcommand{\R}{\mathbb{R}}
\newcommand{\N}{\mathbb{N}}

\newcommand{\C}{\mathbb{C}}

\usepackage{microtype}
\setlist[enumerate,1]{label=\textnormal{(\emph{\roman*})}}

\theoremstyle{thmstyleone}%
\newtheorem{theorem}{Theorem}

\newtheorem{proposition}[theorem]{Proposition}
\newtheorem{corollary}[theorem]{Corollary}

\theoremstyle{thmstyletwo}%
\newtheorem{remark}{Remark}%

\theoremstyle{thmstylethree}%

\theoremstyle{definition}
\newtheorem{assumption}{Assumption}

\usepackage{xurl}

\begin{document}

\title[A model for mosquito-borne epidemic outbreaks with information-dependent protective behaviour]{A model for mosquito-borne epidemic outbreaks with information-dependent protective behaviour}

\author*[1]{\fnm{Simone} \sur{De Reggi}}\email{simone.dereggi@unitn.it}


\author[1]{\fnm{Andrea} \sur{Pugliese}}\email{andrea.pugliese@unitn.it}


\author[1]{\fnm{Mattia} \sur{Sensi}}\email{mattia.sensi@unitn.it}


\author[1]{\fnm{Cinzia} \sur{Soresina}}\email{cinzia.soresina@unitn.it}

\affil*[1]{\orgdiv{Department of Mathematics}, \orgname{University of Trento}, \orgaddress{\street{Via Sommarive 14}, \city{Povo}, \postcode{38123}, \state{Trento}, \country{Italy}}}

\abstract{We investigate a model for a mosquito-borne epidemic in which human hosts may adopt protective behaviour against vector bites in response to information on both past and current disease prevalence. Assuming that mosquitoes can also feed on non-competent hosts (i.e.\ hosts that do not contribute to disease transmission), we first revisit existing results and show that behaviour-driven protection may either decrease or increase the basic reproduction number, depending on the interaction between behavioural response, host composition, and transmission parameters. Assuming that opinion dynamics evolves on a much faster time scale than disease transmission, we then apply Geometric Singular Perturbation Theory to effectively reduce the original two-group model to a model for a homogeneous host population. The reduced system enables a detailed investigation of the impact of information-induced behavioural changes on the transient dynamics of the epidemic, including scenarios in which protective measures lead to outbreaks with low attack rates. Our analysis shows that behavioural responses may either facilitate epidemic control or prolong disease persistence, potentially generating recurrent damped epidemic waves. Numerical simulations are provided to illustrate and support the analytical findings.}

\keywords{behavioural epidemiology, vector-borne diseases, host heterogeneity, Geometric Singular Perturbation Theory, information index}


\pacs[MSC Classification]{34E13, 34E15, 34C60, 37N25, 92D30}

\maketitle

\section{Introduction}
In recent years, there has been growing concern over the rise in mosquito-borne diseases in both endemic and non-endemic areas~\citep{ECDCmosquitoborne, WHOVBD}. 
Malaria, dengue, Chikungunya, Japanese encephalitis, West Nile virus, yellow fever, and Zika are among the most prominent examples of mosquito-borne infections. These diseases are transmitted to humans through the bites of infected mosquitoes carrying pathogens (viruses or other parasites), and represent major threats to global public health~\citep{WHOVBD}. Factors such as climate change, globalisation, and urbanisation have strongly favoured the spread of many of these infections, such as dengue and Chikungunya, to non-tropical areas, creating numerous new habitats suitable for mosquito life~\citep{chala2021emerging}.

Mathematical models play a crucial role in understanding the spread of mosquito-borne infections and in assessing the potential impact of control strategies. The classical Ross--Macdonald framework provides the foundation for much of the mathematical modelling of such diseases; see, for instance, the recent review \cite{PuglieseReview}. Originally proposed by Ronald Ross between 1908 and 1911 to investigate the effect of control measures on malaria transmission~\citep{ross1911}, the model was later extended by George Macdonald to incorporate, among other features, incubation periods, i.e.\ time delays between infection and infectiousness \citep{macdonal1957} (the role of delays had already been considered by \cite{Sharpe1923}; however, they neglected mosquito mortality during the incubation period. Moreover, Macdonald acknowledged \cite{armitage1953} for mathematical ideas related to delays; see \cite{PuglieseReview, smith2012}).
In particular, the mathematical modelling of dengue, which strongly motivates the present work, poses significant epidemiological and methodological challenges; see the recent review~\cite{aguiar2022}.

A key feature distinguishing vector-borne infections from directly transmitted ones is that epidemic containment can be achieved through interventions targeting the vector population~\citep{ogunlade2023}. This insight was already recognised by Ross in his pioneering work on malaria, where he developed arguments that anticipated the modern concept of the \emph{Basic Reproduction Number} (BRN), $R_0$; see~\cite{bacaer2011short, Heesterbeek2002, PuglieseReview}. However, effective control of mosquito-borne epidemics generally requires combining vector-targeted interventions with preventive measures adopted by humans, such as the use of repellents or bed nets. 

Theoretical investigations of the effectiveness of such strategies require mathematical models incorporating behavioural components. As clearly illustrated by the COVID-19 pandemic, human behaviour can substantially influence epidemic spread, since individuals may dynamically adjust their actions in response to available information. For example, individuals may decide whether to be vaccinated or to adopt self-protective measures, either in accordance with public health guidelines or on their own initiative.
The \emph{behavioural epidemiology} of infectious diseases addresses these aspects by incorporating behavioural dynamics into mathematical models of disease transmission~\citep{manfredi2013book}. Its origins can be traced back to the seminal work of~\cite{capasso1978}, who first proposed an SIR (Susceptible--Infectious--Removed) model in which the transmission rate was assumed to decrease as a function of the current \emph{prevalence}, that is, the total number of infected individuals~\citep{dOnofrioManfredi2009}. Since then, numerous approaches have been developed to incorporate behavioural feedback into epidemic models. We refer to the monograph ~\cite{manfredi2013book} and to the recent reviews~\cite{BedsonIntegrated, FunkReview, WangStatistical} for comprehensive overviews of this field.

Although a substantial body of literature now addresses the interplay between the spread of directly transmitted infections 
and human behaviour, comparatively less attention has been devoted to information-induced behavioural effects in vector-borne diseases. In the following paragraphs, we briefly review selected contributions in this area.

For directly transmitted infections, the adoption of preventive behaviour by a subset of individuals generally reduces contact rates across the population, thereby producing an overall beneficial effect~\citep{PuglieseReview}. Mathematically, this is typically modelled by assuming that transmission rates decrease as functions of the level of awareness within the population. In contrast, for vector-borne infections, modelling partial protective behaviour among hosts is far from straightforward, as predictions regarding its effectiveness critically depend on the assumptions governing vector--host ecological interactions~\citep{mccallum2001should, thongsripong2021}. In particular, different assumptions on how mosquito biting rates depend on host availability may lead to substantially different predictions concerning the probability of epidemic invasion and the effectiveness of control strategies~\citep{demers2018dynamic}; see also~\cite{wonham2006transmission}.

As an example, for vector-borne infections, the adoption of self-protective measures by a subset of individuals may divert mosquito bites towards unprotected hosts~\citep{killeen2007exploring, moore2007mosquitoes}, thereby increasing the likelihood of mosquito--host--mosquito transmission cycles~\citep{demers2018dynamic, killeen2007exploring, miller2016risk}. More generally, partially adopted protective behaviour may enhance host heterogeneity, as individuals using repellents or bed nets become less attractive to mosquitoes than those who remain unprotected~\citep{killeen2007exploring, miller2016risk}.

\cite{dye1986} showed that host heterogeneity in vector-borne epidemics may, perhaps counter-intuitively, increase the value of $R_0$, thereby enhancing the invasion potential of the disease~\citep{PuglieseReview}. At the same time, they demonstrated that the final epidemic size may decrease under such heterogeneity~\citep{dye1988}.
This phenomenon was revisited by \cite{miller2013effects}, who considered a simple vector--host model with two host classes and investigated the effects of host diversity, transmission competence, and vector preference on~$R_0$. In a subsequent work, \cite{miller2016risk} analysed a one-vector--one-host model in which the host population is statically divided into two subgroups: individuals adopting protective measures and individuals who do not. 
They found that, depending on parameter values, the use of protective measures by only a fraction of the host population may increase $R_0$, thereby exacerbating the risk of epidemic invasion. On the other hand, if protective measures such as bed nets or repellents significantly increase the time mosquitoes spend attempting to bite protected individuals, then a beneficial effect may also extend to the unprotected subgroup.

Recently, \cite{demers2018dynamic} investigated, primarily through numerical simulations and analysis of the basic reproduction number, the potential effects of behavioural changes during the early phase of a mosquito-borne outbreak, including limiting cases of very small or very large behavioural switching rates.  
The authors showed numerically
that incorporating behavioural changes can lead to a lower $R_0$ compared with a model with static behaviour. 
Building on this framework, subsequent studies \citep{cruz2021,roosa2022} have proposed models of increasing complexity to explore more elaborate scenarios involving behavioural adaptation. 
However, most of these contributions rely predominantly on numerical simulations and provide limited analytical insight into the impact of (possibly information-induced) protective behaviour on epidemic spread, beyond considerations based on the basic reproduction number.

The idea of information-dependent behavioural switching in vector--host models has also been explored in~\cite{misra2013mathematical} and, more recently, in~\cite{hu2023stability},
where it is assumed that only the susceptible host subpopulation is divided into protected (with perfect protection) and unprotected individuals, with behavioural changes driven by exponentially waning information on infection prevalence. Incorporating host demography and thus focusing on an endemic setting, the authors showed that an endemic equilibrium exists whenever $R_0>1$ and may become unstable via a Hopf bifurcation under suitable conditions, leading to the emergence of periodic solutions. The possibility of oscillatory, or even more complex, dynamics in epidemic models with information-induced behavioural responses has been widely discussed in the literature; see, for instance,~\cite{dOnofrioManfredi2009, d2022behavioral, manfredi2013book, poletti2009spontaneous, zhang2023renewal}.

In the literature, it is usual to assume that epidemiological dynamics evolves either much more slowly or much more rapidly than behavioural changes in the host population; see, for example,~\cite{della2024geometric, poletti2009spontaneous}. When such models are formulated as systems of Ordinary Differential Equations (ODEs), a powerful framework for analysing the qualitative behaviour under time-scale separation is the so-called \emph{Geometric Singular Perturbation Theory} (GSPT), originating from the seminal work of Neil Fenichel~\citep{fenichel}. GSPT has been widely applied to the study of natural systems characterised by interacting mechanisms evolving on distinct time scales; see, in particular,~\cite{hek2010geometric,kuehn} for biologically motivated examples. Under the assumption of separated time scales, perturbed ODE systems can be analysed in suitable singular limits that describe the fast and slow dynamics separately, and the resulting reduced problems provide insight into the behaviour of the full system. This approach has been employed in several epidemic models~\citep{jardon2021geometric,jardon2021geometric2,kaklamanos2024geometric}, especially in settings where information (as well as misinformation and opinions) spreads faster than the disease itself~\citep{bulai2024geometric,della2024geometric,schecter2021geometric}. It is worth noting that behavioural and epidemiological processes are not the only mechanisms operating on different time scales. In mosquito-borne diseases, demographic processes in the vector population typically occur on the same time scale as epidemic transmission and therefore cannot be neglected, whereas demographic changes in the human population, as well as waning disease-induced immunity, often evolve much more slowly and may be disregarded in early outbreak scenarios.

Motivated by the recent rise of autochthonous dengue cases in Europe~\citep{ECDCdengue}, we propose a model for a mosquito-borne outbreak in which host individuals may adopt protective behaviour in response to the information available on the current state of the epidemic. In contrast to~\cite{misra2013mathematical, hu2023stability}, we focus on an outbreak scenario and therefore neglect host demography, waning immunity, and seasonal effects in the vector population. Following~\cite{miller2016risk}, we assume that the host population is divided into protected and unprotected subclasses, with individuals dynamically switching behaviour~\cite{demers2018dynamic} at rates depending on epidemic information. To model this mechanism, we adopt the \emph{information index} approach introduced by d'Onofrio, Manfredi, and collaborators, whereby the delayed influence of past infections on current perception is represented through a memory kernel, which may be either distributed or concentrated~\citep{ando2020}. This information may depend on prevalence~\citep{dOnofrioManfredi2009, d2007vaccinating} or on \emph{incidence}, namely the number of new cases per unit time~\citep{d2022behavioral}. Consistently, with the structure of the resulting ODE system, we assume that the human population is closed and constant in size, with no births or deaths. Furthermore, we allow vectors to feed on a non-human host population that does not contribute to transmission~\citep{esteva1998, nishiura2006}. For instance, it has been found that \emph{Aedes albopictus} (main vector of dengue and Chikungunya in Europe) feeds predominantly on humans in urban settings (79--96\% of blood meals), whereas in rural environments humans account for only 23--55\%, with horses and bovines representing common alternative hosts~\citep{valerio2010}.

The paper is organised as follows. In Section~\ref{sec:model}, we introduce the baseline vector-borne model without information and analyse its main properties and asymptotic behaviour. Section~\ref{sec:model+info} incorporates information-induced behavioural changes through an information index based on a memory kernel; specialising to Erlang-distributed kernels, we derive the corresponding reproduction number. In Section~\ref{sec:GSPT1}, we further assume a separation of time scales and apply GSPT to investigate the transient dynamics of the system and present numerical simulations illustrating the analytical results in epidemiologically relevant scenarios. In Section~\ref{Sec:lowattackratio}, we examine the model towards the end of a first outbreak, thus assuming a low attack rate. We investigate,  for selected memory kernels, the existence and stability of equilibria of the reduced system; these correspond to a slow decline of incidence in the original model. Finally, Section~\ref{sec:conclusions} summarises the main findings and outlines possible extensions of the model.

\section{Vector-host model with static protective behaviour}\label{sec:model}

In this section, we introduce the baseline model (a slight extension of the Ross--Macdonald framework incorporating non-human hosts). This model will later be extended to account for a division of the human population into individuals who adopt protective measures against mosquito bites and those who do not. At this stage, behaviour is assumed to be fixed. In the following section, we will allow individuals, irrespective of their epidemiological status, to switch dynamically between behavioural states.

Let $H$ denote the total human population size. We assume $H$ is constant over time (a closed population with no births or deaths). We also assume that the vector can take blood meals from a non-human host population whose individuals do not participate in transmission, i.e. a bite on a non-human host is ``wasted'' from the standpoint of infection spread. 
We will refer to $L$ as the ``effective'' population size of ``non-competent'' hosts; thus, $L/(H+L)$ is the probability that a mosquito bite falls on a ``non-competent'' host, and $H/(H+L)$ is the probability that it falls on a human host.

As for the epidemiological dynamics, we assume an SIR--SI framework. Humans follow an SIR model without demography, i.e. we assume no births or deaths, and permanent immunity after recovery. Mosquitoes follow an SI model with demography since their average lifespan is orders of magnitude shorter than that of humans.
For simplicity, we neglect the incubation periods in both mosquitoes and hosts, during which individuals are infected but not yet infectious. It is well-known that including the extrinsic incubation period (the one in mosquitoes) reduces the value of $R_0$, while having only minor effects on the qualitative dynamics of the system~\citep{macdonal1957,PuglieseReview}.

Thus, in the absence of protective behaviour, the dynamics is described by the following system of ODEs~\citep{esteva1998, nishiura2006}: 
\begin{equation}\label{SIR-SI}
	\left\{\setlength\arraycolsep{0.1em}
	\begin{array}{rl} 
		S_H' &= -bp_{H\leftarrow M}I_M \cfrac{S_H}{H+L},\\[3mm]
		I_H' &= bp_{H\leftarrow M}I_M \cfrac{S_{H} }{H+L} - \gamma  I_{H} ,\\[5mm]
		R_H'&=\gamma  I_H\\[3mm]
		S_M' &= \Lambda-bp_{M\leftarrow H}S_M \cfrac{I_{H} }{H+L}-\mu S_M,\\[3mm]
		I_M' &= bp_{M\leftarrow H}S_M \cfrac{I_{H}}{H+L} -\mu  I_M,
	\end{array} 
	\right.
\end{equation}
where we denote with $S_H(t),\, I_H(t),\, R_H(t)$ the number of human individuals who are susceptible, infectious, and removed, respectively, at time $t\ge 0$, and with $S_M(t),\, I_M(t)$ the number of mosquitoes who are susceptible and infectious, respectively, at time $t\ge 0$. 
The parameter $b$ is the constant per-capita mosquito biting rate, assumed to be independent of the mosquito’s infection status. The quantities~$p_{H\leftarrow M},\, p_{M \leftarrow H}\in [0, 1]$ are the probabilities of transmission per bite from mosquitoes to humans and from humans to mosquitoes, respectively.
The parameter $\gamma>0$ is the per-capita human recovery rate (so that the average infectious period is $1/\gamma$). The parameters $\Lambda$ and $\mu$ are the mosquito recruitment rate and per-capita death rate, respectively (hence the average mosquito lifespan is $1/\mu$). 

For the mosquito population dynamics, we assume for simplicity that $\Lambda = \mu M$. This implies~$M'=S_M'+I_M'=0$, so that the total mosquito population $M$ remains constant.
Similarly,~$H'=S'_H+I'_H+R'_H=0$. Hence, the equation for $R_H$ can be omitted since $R_H(t)\equiv H-S_H(t)-I_H(t)$, for all $t\ge 0$.

\begin{remark}\label{rem:other_L}
	Recall that $L$ represents non-human targets for mosquito bites.
	An alternative interpretation of the constant $L$ is a ``penalty term'' in the expression for the mean time mosquitoes spend searching for human hosts; for instance, it may reflect the spatial separation among humans. Then, by defining $\zeta\coloneqq L^{-1}$ and 
	$$\hat b\coloneqq \cfrac{b H}{H+L}=\cfrac{b\zeta H}{1+\zeta H},$$ 
	we can interpret $\hat b$ as the effective mosquito biting rate, determined by a Holling-type II functional response, as in~\cite{demers2018dynamic, miller2016risk, yakob2016biting}. 
	The force of infection acting on humans can then be written as
	\begin{equation*}
		\lambda_{H\leftarrow M}\coloneqq  \hat b\ \cfrac{I_M}{H}.
	\end{equation*}
	In summary, $L$ can be viewed as accounting for the effect on the average mosquito searching time of both the presence of non-competent hosts and the spatial distribution (or mutual distance) of humans. 
\end{remark}

\subsection{Protective behaviour}
To incorporate protective behaviour into model \eqref{SIR-SI}, following recent works in the literature~\citep{demers2018dynamic, miller2016risk, roosa2022}, we assume that the host population $H$ can be partitioned as $H=H_{P}+H_{NP}$. Here, $H_{P}=pH$, with $p\in [0,1]$, represents the subpopulation of \emph{protected} individuals, i.e. those adopting protective behaviour, while $H_{NP}=(1-p)H$ represents those who do not (the subpopulation of \emph{non-protected} individuals). 

To model the efficacy of protective behaviour, we assume that protected individuals are relatively less exposed to mosquito bites compared to non-protected individuals. Let $q\in [0, 1]$ denote the probability of protection failure (see, for instance,~\cite{roosa2022}). The probability that a mosquito bites a protected individual in the group $H_{P}$ is
\begin{equation}\label{VBsystem4}
	P_b^{H_P}(p,q)=\frac{q H_{P}}{q H_{P}+H_{NP}+L}=\frac{pqH}{c(p, q)H+L},
\end{equation}
where 
\begin{equation}\label{c(p,q)}
	c(p,q)\coloneqq 1-p(1-q).    
\end{equation}
Note that 
$P_b^{H_P}(p,q)< P_b^{H_P}(p,1)$, for $p\in (0,1),\ q\in[0,1)$,
i.e. the probability that an individual is bitten by a mosquito is effectively smaller if they adopt effective protection measures. Conversely, for the probability that a mosquito bites a non-protected individual in the group $H_{NP}$ or a non-human target $L$, we have
\begin{equation*}
	P_b^{H_{NP}}(p,q)\coloneqq \frac{(1-p)H}{c(p,q)H+L}>P_b^{H_{NP}}(p,1),\quad\text{for}\quad p\in (0,1),\ q\in[0,1),
\end{equation*}
and
\begin{equation*}
	\frac{L}{c(p,q)H+L}> \frac{L}{H+L}\quad\text{for}\quad p\in (0,1),\ q\in[0,1),
\end{equation*}
respectively. Then, taking human protective behaviour into account, the model~\eqref{SIR-SI} becomes
\begin{equation}\label{VHBnumbers}
	\left\{\setlength\arraycolsep{0.1em}
	\begin{array}{rl} 
		S_{P}' &= -\beta_{H\leftarrow M}I_M \cfrac{q S_{P} }{c(p,q)H+L},\\[3mm]
		I_{P}' &= \beta_{H\leftarrow M}I_M \cfrac{q S_{P} }{c(p,q)H+L} - \gamma  I_{P} ,\\[3mm]
		S_{NP}' &=  -\beta_{H\leftarrow M}I_M \cfrac{S_{NP} }{c(p,q)H+L},\\[3mm]
		I_{NP}' &= \beta_{H\leftarrow M}I_M \cfrac{S_{NP} }{c(p,q)H+L} - \gamma  I_{NP},\\[3mm]
		I_M' &= \beta_{M\leftarrow H}(M-I_M) \cfrac{q I_{P} +I_{NP} }{c(p,q)H+L} -\mu I_M,
	\end{array} 
	\right.
\end{equation}
where now $S_P, I_P$ denote the susceptible and infectious individuals who adopt protective behaviour, while $S_{NP}, I_{NP}$ denote the susceptible and infectious individuals who do not. Moreover, 
$$\beta_{H\leftarrow M}\coloneqq b p_{H\leftarrow M}, \qquad \beta_{M\leftarrow H}\coloneqq b p_{M\leftarrow H},$$
where $b$ is the biting rate and $p_{H\leftarrow M},\,p_{M\leftarrow H}$ the transmission probabilities per bite.

\begin{remark}\label{remark_no_handling}
	In system~\eqref{VHBnumbers}, 
	we assume that protective behaviour does not substantially alter the mosquito handling time; that is, the time required for a mosquito to attempt or complete a bite is assumed to be the same for protected and non-protected individuals. Consequently, protection affects only the probability of successful biting (through the factor $q$), and not the biting rate itself.
\end{remark}
\begin{remark}
	Consider the case $L=0$.  
	We note that 
	neither of the limits
	$$\lim_{(p, q)\to (1^-, 0^+)}P_b^{H_P}(p, q) \;\;\hbox{ and } \;\;\lim_{(p, q)\to (1^-, 0^+)}P_b^{H_{NP}}(p, q)$$
	exists.
	This ``inconsistency'' in the model is expected, as we are assuming that the number of available hosts is zero while the mosquito biting rate remains fixed; that is, mosquitoes are assumed to bite a fixed number of hosts per unit time regardless of host availability. Clearly, this situation cannot occur when~$L>0$.
\end{remark}

In the following, we argue in terms of the fractions $S_P/H$, $I_P/H$, $S_{NP}/H$, $I_{NP}/H$, and $I_M/M$, which, for later convenience, will still be denoted by the same variables that have been used so far.
The resulting model reads
\begin{equation}\label{VBH}
	\left\{\setlength\arraycolsep{0.1em}
	\begin{array}{rl} 
		S_{P}' &= -\beta_{H\leftarrow M}\rho I_M \cfrac{q S_{P} }{c(p,q)+l},\\[3mm]
		I_{P}' &= \beta_{H\leftarrow M}\rho I_M \cfrac{q S_{P} }{c(p,q)+l} - \gamma  I_{P} ,\\[3mm]
		S_{NP}' &=  -\beta_{H\leftarrow M}\rho I_M \cfrac{S_{NP} }{c(p,q)+l},\\[3mm]
		I_{NP}' &= \beta_{H\leftarrow M}\rho I_M \cfrac{S_{NP} }{c(p,q)+l} - \gamma  I_{NP},\\[3mm]
		I_M' &= \beta_{M\leftarrow H}(1-I_M) \cfrac{q I_{P} +I_{NP} }{c(p,q)+l} -\mu I_M,
	\end{array} 
	\right.
\end{equation}
where now
\begin{equation}\label{ratiosLM}
	l\coloneqq \frac{L}{H},\quad\text{and}\quad \rho\coloneqq \frac{M}{H}.    
\end{equation}
For the reader's convenience, we report the model parameters and their interpretation in \Cref{TableSIR-SI}, while $c(p,q)$ is given by \eqref{c(p,q)}. In \Cref{splitSIRSI}, we represent the dynamics of model \eqref{VBH}.

\begin{table}
	\caption{Parameters appearing in model \eqref{VBH} and their interpretations.}
	\label{TableSIR-SI}
	\begin{tabular}{c  c}
		\toprule
		Parameter & Interpretation \\[0.5mm]
		\midrule
		$l$ & ratio between non-competent hosts and human population sizes \\[0.5mm]
		$\rho$ & ratio between mosquito and human population sizes\\[0.5mm]
		$p$ & fraction of protected humans \\[0.5mm]
		$q$ & probability of protection failure \\[0.5mm]
		$1/\gamma$ & average human infectious period \\[0.5mm]
		$1/\mu$ & average mosquito life span \\[0.5mm]
		\bottomrule
	\end{tabular}
\end{table}
\begin{figure}
	\centering
	\begin{tikzpicture}[draw, roundnode/.style={circle, draw=black, thin,
			minimum size=1.2cm},align=center, node distance = 2cm and 8cm, > = Stealth, accepting/.style = {accepting by arrow}, accepting distance = 2em, initial distance = 1.5em,   >=latex,shorten >=1pt, shorten <=1pt]
		\node[roundnode
		] (Sp) {$S_{P}$};
		
		\node[roundnode, right = of Sp, accepting right, accepting text = $\gamma I_P$
		] (Ip) {$I_{P}$};
		
		\path[->] 
		(Sp) edge              node[anchor=center, above]                 {$ \frac{q\beta_{H\leftarrow M}\rho I_MS_{P}}{c(p,q)+l}$} (Ip);              
		\node[roundnode, below = of Sp
		] (Snp) {$S_{NP}$};
		
		\node[roundnode,  below = of Ip, accepting right, accepting text = $\gamma I_{NP}$
		] (Inp) {$I_{NP}$};
		
		\coordinate(a) at ($(Sp)!0.5!(Ip)$);
		
		\coordinate(b) at ($(Sp)!0.5!(Snp)$);
		
		\node[roundnode,
		accepting left, initial distance = 2.5em, accepting text = $\mu  I_M$,   
		] (Iv) at (a|-b) {$I_M$};

		\path[->] 
		(Snp) edge              node[anchor=center, below]                 {$\frac{\beta_{H\leftarrow M} \rho I_M S_{NP}}{c(p,q)+l}$} (Inp);   
		
		\coordinate(c) at ($(Ip)!0.5!(Inp)$);
		
		\path[->, dashed] 
		(Ip) edge              node[anchor=center, left]                 {} (c);              
		\path[->, dashed] 
		(Inp) edge              node[anchor=center, left]                 {} (c);    
		
		\path[->] 
		(c)+(0.25,0) edge              node[anchor=center, above]{\textcolor{black}{$\frac{\beta_{M\leftarrow H}(1-I_M)(qI_{P}+I_{NP})}{c(p,q)+l}$}}                  (Iv);  
		
		\coordinate(d) at ($(Sp)!0.5!(Ip)$);
		\coordinate(e) at ($(Snp)!0.5!(Inp)$);
		
		\path[->, dashed] 
		(Iv) edge              node[anchor=center, above]{} (d);
		
		\path[->, dashed] 
		(Iv) edge              node[anchor=center, above]{} (e);

	\end{tikzpicture}
	\caption{Flow chart for system \eqref{VBH}. Straight lines: compartmental movements within each population; dashed lines: infections between populations (mosquitoes infecting humans and vice versa).\label{splitSIRSI}}
\end{figure}

As we are dealing with fractions, the biologically meaningful state space for the solutions of model \eqref{VBH} is the set
\begin{align*}
	\Omega:=\{(&S_P, I_P, S_{NP}, I_{NP}, I_M) \in \R^5_{\ge 0}\ :\\
	&0\le S_P, I_P, S_{NP}, I_{NP}, I_M;\ S_P+ S_{NP},\ I_P+ I_{NP},\ I_M\le 1\}, 
\end{align*}
for which the following classical result holds.

\begin{proposition}
	For every initial condition $X(0)\coloneqq(S_P(0), I_P(0), S_{NP}(0), I_{NP}(0), I_M(0)) \in \Omega$, the system \eqref{VBH} admits a unique solution $X(t)\coloneqq(S_P(t), I_P(t), S_{NP}(t), I_{NP}(t), I_M(t)) \in \Omega$ that is globally defined in the future. 
	Furthermore, the following inequalities hold for every $t\ge 0$: 
	\begin{equation}\label{ineqSPSNP}
		S_P(t)+I_P(t)\leq S_P(0)+I_P(0),\quad\text{and}\quad S_{NP}(t)+I_{NP}(t)\leq S_{NP}(0)+I_{NP}(0).
	\end{equation}
\end{proposition} 
\begin{proof}
	Let $X(0)\in \Omega$. The local existence and uniqueness of the solution $X(t)$ to the Cauchy problem associated with \eqref{VBH} follow from the Cauchy--Lipschitz theorem. 
	Now, for each $Y\in\{ S_P, I_P, S_{NP}, I_{NP}, I_M\}$, one has
	$Y'|_{Y=0}\geq 0$, which implies that $X(t)\in \R^5_{\ge 0}$ for all $t\ge 0$ for which the solution is defined. Furthermore, the inequalities 
	$(S_j+I_j)'=-\gamma I_j\leq0,\quad j\in\{P,\ NP\}$,
	and $I'_M|_{I_M\ge 1}  \le -\mu I_M\leq 0$ hold. Hence, using standard arguments, one obtains that $X(t)\in \Omega$ is globally defined in the future, and 
	the inequalities in \eqref{ineqSPSNP} hold.
\end{proof}

\subsection{Control reproduction number}
In epidemiological models, the BRN $R_0$ is the expected number of secondary cases produced by a typical infected individual in a completely susceptible population during its entire infectious period~\citep{diekmann1990}. 
When intervention or control strategies are implemented in an otherwise fully susceptible population, the relevant threshold quantity is referred to as the \emph{Control Reproduction Number} (CRN), denoted by $R_c$. This quantity shares the same threshold property as $R_0$~\citep{pellis2022}.

In this section, we compute the CRN $R_c\coloneqq R_c(p,q)$ for the model with protective behaviour \eqref{VBH}, following the general framework introduced in~\cite{diekmann1990}. In particular, we apply the Next Generation Matrix (NGM) method described in
\cite{diekmann2010construction, van2002}. 
To this end, we first observe that system~\eqref{VBH} admits the DFE
\begin{equation*}
	(S_{P}^*,\ I_{P}^*,\ S_{NP}^*,\ I_{NP}^*,\ I_M^*)\coloneqq (p,\ 0,\ 1-p,\ 0,\ 0),
\end{equation*}
We then focus on the early stage of the epidemic. In particular, we assume that the total human susceptible fraction satisfies~$S_{P}+S_{NP}\approx 1$ and that the mosquito population is almost entirely susceptible, i.e. $S_M=1-I_M\approx 1$. We therefore consider the linearised system for the infected compartments around the DFE
\begin{equation}\label{VBlin}
	\left\{\setlength\arraycolsep{0.1em}
	\begin{array}{rl} 
		I_{P}' &= \beta_{H\leftarrow M}\rho I_M\cfrac{q p}{c(p,q)+l}   - \gamma  I_{P} ,\\[3mm]
		I_{NP}' &= \beta_{H\leftarrow M} \rho I_M\cfrac{1-p}{c(p,q)+l}  - \gamma  I_{NP} ,\\[3mm]
		I_M' &= \beta_{M\leftarrow H}\cfrac{q I_{P} +I_{NP} }{c(p,q)+l} -\mu I_M,
	\end{array} 
	\right.
\end{equation}
where $\rho$ and $l$ are defined as in \eqref{ratiosLM}. Following~\cite{diekmann2010construction, van2002}, we define a matrix accounting for \emph{infection} 
\begin{equation}\label{infection:matrix}
	B(p, q)\coloneqq 
	\begin{pmatrix} 
		0 & 0 & \dfrac{\rho qp\beta_{H\leftarrow M }}{c(p,q)+l}\\[2mm]
		0 & 0 & \dfrac{\rho(1- p)\beta_{H\leftarrow M}}{c(p,q)+l}\\[2mm]
		\dfrac{q\beta_{M\leftarrow H}}{c(p,q)+l} &
		\dfrac{\beta_{M\leftarrow H}}{c(p,q)+l} & 0
	\end{pmatrix},\quad p,q\in[0,1], 
\end{equation}
and a matrix accounting for \emph{transition} $\Sigma\coloneqq \diag(\gamma, \gamma, \mu)$,
so that the Jacobian matrix of system \eqref{VBlin} can be easily written as $B(p,q)-\Sigma$.
Then, according to~\cite{diekmann1990}, the CRN $R_{c}(p, q)$ is obtained as the spectral radius of the NGM 
$K(p, q)\coloneqq B(p, q)\Sigma^{-1}$,
and explicitly reads
\begin{equation}\label{R0het}
	R_c=R_0\;
	\cfrac{1+l}{c(p,q)+l}\;\sqrt{1-p(1-q^2)}\,,
\end{equation}
where
\begin{equation}\label{R0}
	R_{0}\coloneqq \frac{1}{1+l}\; 
	\sqrt{\cfrac{\beta_{H\leftarrow M}\beta_{M\leftarrow H}}{\gamma\mu } \rho}
\end{equation}
is the BRN (namely, the reproduction number in the absence of protective behaviour). 

Note that, if $p=0$ (i.e. no individuals adopt protective behaviour) or $q=1$ (i.e. protection is completely ineffective), then $R_c=R_0$. 
Moreover, it is interesting to observe that if $p=1$ and $l=0$, then $R_c=R_0$.
This is a direct consequence of the Ross--Macdonald-like assumptions underlying the model. Indeed, in this case, all individuals adopt protective behaviour. However, the mosquito biting rate is assumed to be fixed, meaning that mosquitoes feed on a given number of hosts per unit time regardless of host behaviour. Since $l=0$, mosquitoes have no alternative hosts and therefore must feed on humans. Consequently, whenever protection fails, all mosquito bites are effectively concentrated on those individuals for whom protection is unsuccessful. As a result, the overall transmission potential remains unchanged.

It is therefore natural to ask whether the presence of awareness in the human population reduces the reproduction number. In particular, we seek conditions ensuring that
\begin{equation}\label{rless}
	R_c(p, q)< R_{0}.
\end{equation}
In this regard, we have the following result.
\begin{proposition}\label{propineqRatio}
	Let  $p\in (0, 1)$, $q\in [0, 1)$ and $l\ge 0$. Then condition \eqref{rless} holds if and only if
	\begin{equation}\label{threshold}
		l>F(p,q),\qquad\text{with}\qquad F(p,q)\coloneqq \frac{(1-p)(1-q)}{\sqrt{1-p(1-q^2)}+q}\,.
	\end{equation}
\end{proposition}
\begin{proof} 
	From \eqref{R0het}, the inequality $R_c(p, q)< R_{0}$ is equivalent to \begin{equation*}
		\frac{1+l}{(1-p+q p) + l}\sqrt{1-p+q^2 p}<1\iff (1+l)\sqrt{1-p+q^2 p}<(1-p+qp) + l\,.
	\end{equation*}
	Rearranging terms yields
	\begin{equation}\label{ineqHL}
		\left[\sqrt{1-p+q^2 p}-(1-p+q p)\right]+l\left[\sqrt{1-p+q^2 p}-1\right]< 0\,.
	\end{equation}
	We consider the first term on the LHS of \eqref{ineqHL}. 
	Note that 
	$\sqrt{1-p+q^2 p}-(1-p+q p)= \sqrt{1-p(1-q^2)}-[1-p(1-q)]$
	and, since  $1-p(1-q)\ge 0$ for all $p, q\in [0, 1]$, we get 
	$\sqrt{1-p(1-q^2)}-[1-p(1-q)]>0\iff  (1-q)^2p(1-p)>0$, which is obvious if $q\in [0, 1)$ and $p\in (0, 1)$.
	Hence
	\begin{equation*}
		\sqrt{1-p+q^2 p}-(1-p+q p)> 0,\quad p\in (0,1),\quad q\in [0, 1).
	\end{equation*}
	As for the second term of inequality~\eqref{ineqHL}, it is clear that $\sqrt{1-p+q^2 p} = \sqrt{1-p(1-q^2)}<1$, for $p\in (0, 1]$ and $q\in [0, 1).$
	Hence
	\begin{equation}\label{Lless0}
		l\left[\sqrt{1-p+q^2 p}-1\right]< 0,\qquad p\in (0, 1],\quad q\in [0, 1).
	\end{equation}
	This shows that the presence of alternative hosts $l$ always reduces $R_c$ when protective behaviour is present.
	Using \eqref{Lless0}, inequality \eqref{ineqHL} can be rewritten as
	$l[1-\sqrt{1-p+q^2 p}]> \sqrt{1-p+q^2 p}-(1-p+q p)$
	which is equivalent to
	\begin{equation}\label{almostRHS}
		l>\frac{\sqrt{1-p+q^2 p}-(1-p+q p)}{1-\sqrt{1-p+q^2 p}}.
	\end{equation}
	Finally, after algebraic manipulations, the right-hand side of \eqref{almostRHS} can be rewritten, for $p,q\in (0,1)$, as 
	\begin{align*}
		\frac{\sqrt{1-p+q^2 p}-(1-p+q p)}{1-\sqrt{1-p+q^2 p}}
		&=\cfrac{(1-p)(1-q)}{\sqrt{1-p(1-q^2)}+q}\,.
	\end{align*}
\end{proof}

\begin{corollary}
	\label{prop_parabola}
	Let  $q\in [0, 1)$ and $l>0$. Then there exists $\bar p < 1$ such that \eqref{rless} holds if and only if $\bar p < p < 1$. 
	If $l \ge \bar l$, where $\bar l=(1-q)/(1+q)$ is the positive root of $$l^2(1+q)+2ql-(1-q)=0,$$ then $\bar p \le 0$, so that \eqref{rless} holds for all $p \in (0,1)$. On the other hand, if $0 < l < \bar l$, $\bar p \in (0,1)$.
\end{corollary}
\begin{proof}
	Rewrite \eqref{threshold} as $Q(p):=Ap^2+Bp+C<0$, 
	where
	$A\coloneqq (1-q)^2$, $B\coloneqq (1-q)[l^2(1+q)+2ql-2(1-q)]$, and
	$C\coloneqq (1-q)[1-q(1+2l)-l^2(1+q)]$.
	It is easy to see that $Q(1)=0$, hence one root of the quadratic equation is $p=1$. Moreover, $Q'(1) > 0$ for $l > 0$, which implies that the other root $\bar p$ satisfies $\bar p<1$, and therefore $Ap^2+Bp+C<0$ for $(\bar p,1)$. Finally, $\bar p \le 0$ if and only if $Q(0) = C \le 0$, which holds precisely when $l \ge \bar l$.   
\end{proof}

\begin{figure}[t]
	\begin{center}
		\includegraphics[width=0.6\linewidth]{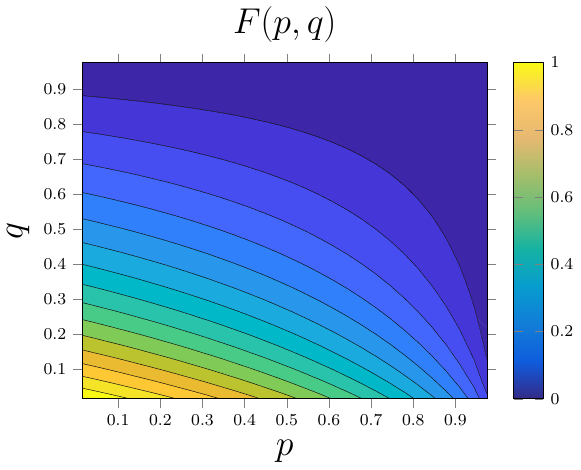}
		\caption{The behaviour of the function $F(p,q)$ defined in \eqref{threshold}. Note that it is independent of all the model parameters except $p$ and $q$. \label{fig:F(p, q)}}
	\end{center}
\end{figure}

The behaviour of the function $F(p,q)$ is illustrated in Figure~\ref{fig:F(p, q)}. Note that it depends only on $p$ and $q$, and is therefore independent of all other model parameters. Essentially, \Cref{propineqRatio} shows that if there are not enough ``other available hosts'', or, under the other interpretation of $L$ (recall Remark \ref{rem:other_L}), if human hosts are not sufficiently distanced from each other, then protective behaviour adopted by only part of the population may actually lead to an increase in the reproduction number $R_0$. 

From \Cref{propineqRatio}, we also immediately obtain the following corollary.
\begin{corollary}\label{corollaryRc}
	Let $l=0$. Then,
	$R_c(p, q)\ge R_{0},\ \forall p\in [0,1),\ q\in (0,1].$
	In particular, defining $r(p, q):=R_c(p, q)/R_0$, we have $r(1, q)=1$,\ $r(p, q)>1$, with
	$\partial_q r(p,q) < 0$ for all $q \in [0,1)$ and $p\in(0,1)$. 
\end{corollary}
Thus, either in the absence of other hosts or under conditions of high population density, static protective behaviour in the human population always increases the reproduction number $R_0$ when it is only partially adopted ($p<1$) or partially effective ($q>0$). 

This is consistent with the results of~\cite{dye1986, dye1988} for vector-host models with heterogeneous host populations. In particular, \cite{dye1986} show that, in the absence of non-competent hosts, static heterogeneity in the host population always leads to an increased value of $R_0$. \cite{dye1988} further discuss the effect of such heterogeneity on the final epidemic size. See also the more recent work by~\cite{bolzoni2015}, where, building on the results of~\cite{dye1986, dye1988}, the authors investigate the role of heterogeneity on the invasion probability of a vector-borne disease in multi-host models.

Finally, it is worth observing from \eqref{R0het} that the presence of other hosts, i.e. $l>0$, always decreases the value of $R_0$ regardless of the level of awareness or intervention, and even in their absence (i.e. when $q=1$ or $p=0$). This phenomenon is commonly known (especially in the context of tick-borne infections) as the \textit{dilution effect}~\citep{rosa2007effects}; see also~\cite{PuglieseReview} for further discussion.
Finally, considering $R_c(p, q)$ as a function of $q$, it is straightforward to derive the following result.
\medskip
\begin{proposition}\label{minimum}
	Let $l > 0$. Then, for every $p\in [0, 1]$, the function $R_c(p, q)$, viewed as a function of $q$, attains a minimum value, smaller than $R_{0}$, at 
	$q= (1-p)/(1-p+l)\,.$
\end{proposition} 
\Cref{minimum} shows that when protective behaviour is static, the largest reduction of $R_0$ is achieved at an intermediate level of protection effectiveness, $1-q$. \Cref{fig:R0qpl} illustrates the behaviour of $R_c$ as a function of $q$ for different values of $p$ and $l$.
\begin{figure}
	\begin{center}
		\includegraphics[width=1\linewidth]{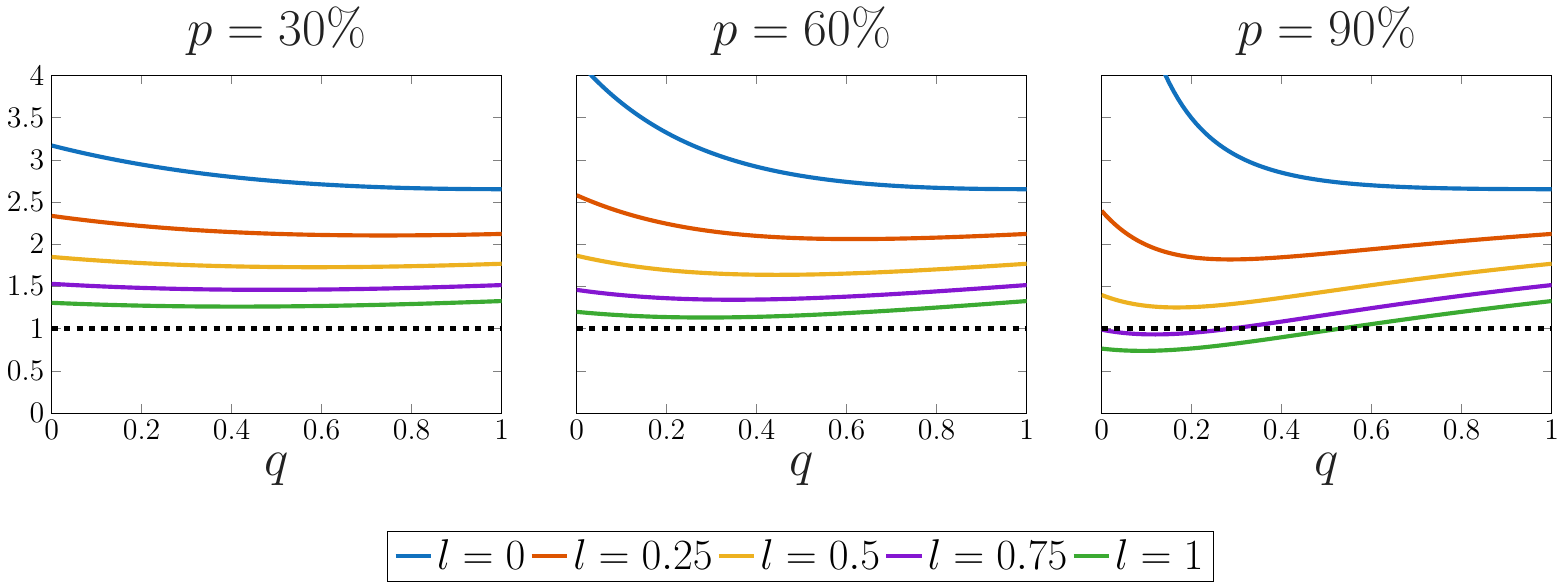}
		\caption{Behaviour of $R_c$ as a function of $q$ for different values of $p$ and $l$, with $\rho=2$ and epidemiological parameters as in \Cref{TableValues}. }\label{fig:R0qpl}
	\end{center}
\end{figure}

\subsection{Asymptotic dynamics and final size}
In this section, we study the long-term dynamics of the solutions of system \eqref{VBH} and we numerically investigate the effect of static protective behaviour in the host population on the final epidemic size~\citep{brauer2017, brauer2019, diekmann2013mathematical}. 
We begin by observing that, since the model does not include human demography or loss of immunity, the only non-trivial equilibria of \eqref{VBH} are those satisfying $I_P=I_{NP}=I_M=0$. Consequently, after an outbreak, the disease is expected to die out in the long run. This result is established in the following proposition. 

\begin{proposition}\label{prop:extinction}
	For system \eqref{VBH},  
	$\lim_{t\to+\infty} I_{P}(t) = \lim_{t\to+\infty} I_{NP}(t) =\lim_{t\to+\infty} I_M(t) =0.$
\end{proposition}
\begin{proof}
	For $S_{P},\, S_{NP},\, I_{P},\, I_{NP}\, \in \R_{> 0}$, the inequalities in \eqref{ineqSPSNP} hold. 
	Hence, both $X_P\coloneqq S_{P}+I_{P}$ and $X_{NP}\coloneqq S_{NP}+I_{NP}$ are decreasing functions of $t$, thus they admit limits $X_{P}^ \infty,\, X_{NP}^\infty\ge 0$, respectively, for $t\to +\infty$. 
	Then, we observe that
	\begin{equation*}
		-\infty< X_P^\infty-X_P(0) =-\gamma \int_0^{+\infty} I_{P}(t)\dd t,
	\end{equation*}
	which implies that $I_{P}(t)\to 0$ as $t\to+\infty$. The claim follows by applying the same reasoning to $I_{NP}$, and by noticing that $I_{P}(t),I_{NP}(t)\to 0$ implies $I_M(t)\to 0$.
\end{proof}

In the remainder of this section, we investigate how partial protective behaviour affects the final epidemic size. Note that, in general, deriving final size relations even for simple vector-borne epidemic models is not straightforward. Indeed, to the best of our knowledge, these aspects have only been investigated in a few recent papers by~\cite{brauer2017, brauer2019} (see also~\cite{gimenez2023final, tsubouchi2019calculation}), where the author derives approximate formulas that provide lower and upper bounds. See also~\cite{PuglieseReview} for further discussion. 

In this section, building on the discussion above, we use numerical simulations to explore the effect of protective behaviour on the final epidemic size.
In particular, we plot $Y_P,\, Y_{NP}$, together with their sum $Y_H:=Y_P+Y_{NP}$, as well as the fractions within each subgroup, $Y_P/p$ and $Y_{NP}/(1-p)$, for $Y\in \{I, R\}$. We then consider the (numerically approximated) limits 
$$R^\infty_H:=\lim_{t\to+\infty} R_H(t),\quad \hat R^\infty_P:=\lim_{t\to+\infty} R_P(t)/p\quad \hbox{ and }\quad \hat R^\infty_{NP}:=\lim_{t\to+\infty} R_{NP}(t)/(1-p),$$
which represent the total infected fraction ($R^\infty_H$) and the infected fractions within each subgroup ($\hat R^\infty_{NP}$ and $\hat R^\infty_{NP}$), respectively.

\begin{table} 
	\caption{Parameter values inspired by dengue fever and taken from~\cite{haonan2025e}. Recall that $\beta_{H\leftarrow M}=b\ p_{H\leftarrow M}$ and $\beta_{M\leftarrow H}=b\ p_{M\leftarrow H}$ \eqref{VBH}.}
	\label{TableValues}
	\begin{tabular}{c c c c}
		\toprule
		Parameter &  Interpretation & Value & Reference \\[0.5mm]
		\midrule
		$b$ & mosquito biting rate & 0.5 days$^{-1}$ &~\cite{aguiar2022}\\[0.5mm]
		$p_{H\leftarrow M}$ & probability of transmission per bite from mosquito to human & 0.75 &~\cite{newton1992model} \\[0.5mm]
		$p_{M\leftarrow H}$ & probability of transmission per bite from human to mosquito & 0.375 &~\cite{newton1992model} \\[0.5mm]
		$1/\gamma$ & average human infectious period & 5 days & ~\cite{newton1992model} \\[0.5mm]
		$1/\mu$ & average mosquito life span & 10 days &~\cite{alphey2011model}\\[0.5mm]
		\bottomrule
	\end{tabular}
\end{table}

We use model parameters inspired by dengue transmission, taken from~\cite{haonan2025e} and reported in \Cref{TableValues} with their original interpretation. From these parameters, we compute $\beta_{H\leftarrow M} = b p_{H\leftarrow M},$ $\beta_{M\leftarrow H} = b p_{M\leftarrow H}$, together with the values of
$\gamma$ and $\mu$. As for $l$, in Figure \ref{fig:static1} 
we take $l=0.25$, so that $H$ represents the $80\%$ of the total host population $H+L$, while in Figure \ref{fig:static3}, we set $l=1$ so that $H$ represents the 50\% of $H+L$. In both cases, for the sake of comparison, $\rho$ is chosen so that the value $R_0\approx 2.12$ is preserved. 
Finally, we assume that the epidemic starts with a very small fraction of infected mosquitoes and no infected humans. In particular, we consider as initial conditions $S_P(0)=p$, $S_{NP}(0)=1-p$, $I_P(0)=I_{NP}(0)=0$, and $I_M(0)=10^{-4}$.
\begin{figure} 
	\begin{center}
		\includegraphics[width=1.\linewidth]{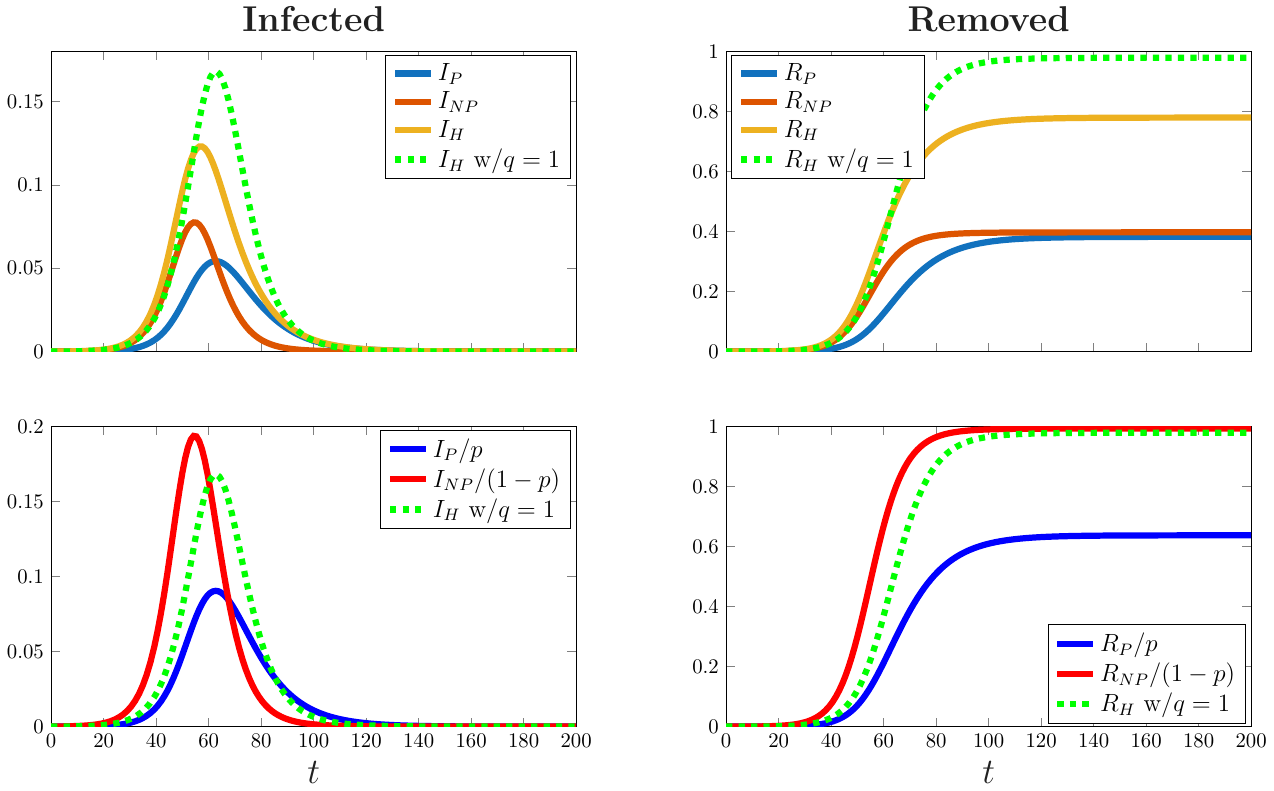}
		\caption{Simulation of model \eqref{VBH} with $l= 0.25$, $\rho=2$, epidemiological parameters as in \Cref{TableValues}, $q = 0.2$ and $p=0.6$, yielding $R_0\approx 2.12$ and $R_c\approx 2.24$. 
			Top row: fraction of protected individuals (solid blue), fraction of non-protected individuals (solid red), total fraction of individuals (solid yellow), and the corresponding trajectory for the model without protection (namely, with $q=1$, dashed green), shown for the infected class (left) and the removed class (right).
			Bottom row: same as fractions of individuals within each subgroup. Left: infected fractions. Right: removed fractions.
		}\label{fig:static1}
	\end{center}
\end{figure}

\begin{figure} 
	\begin{center}
		\includegraphics[width=1.\linewidth]{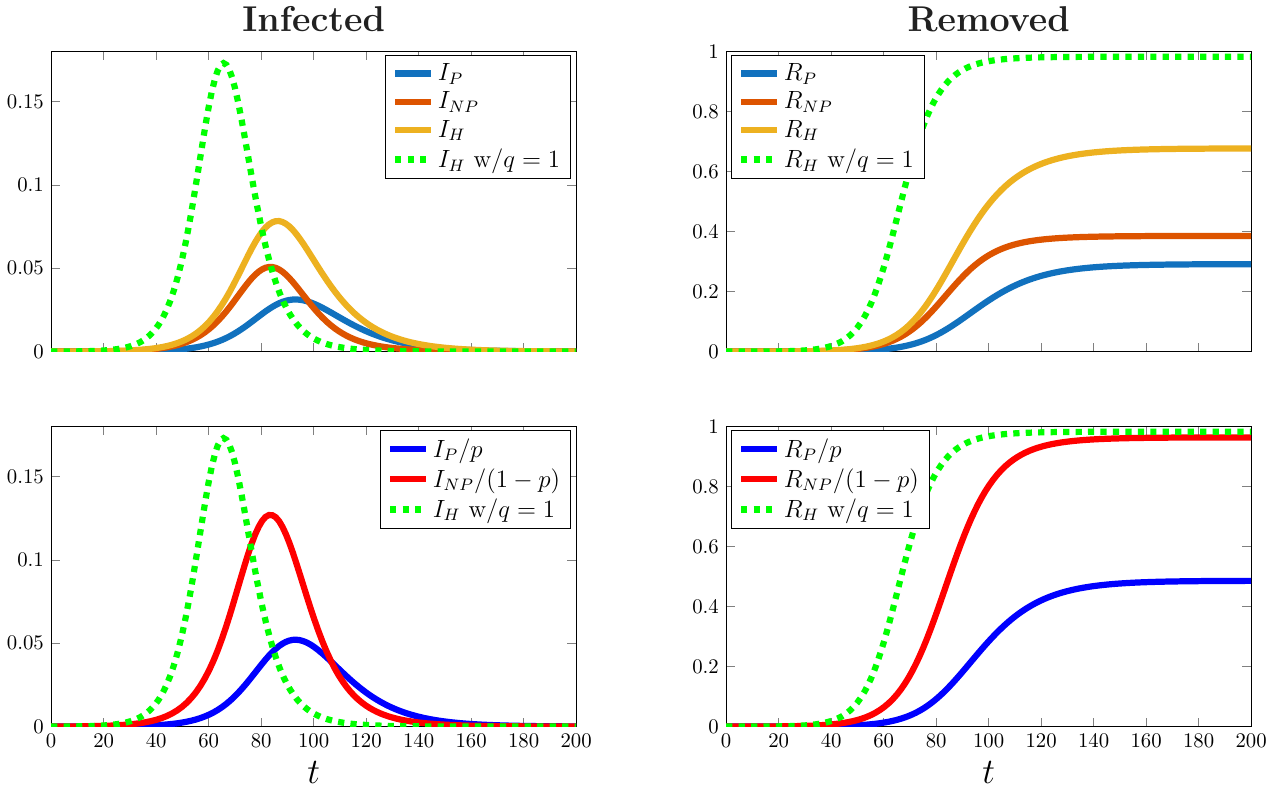}
		\caption{Simulation of model \eqref{VBH} with parameters and initial conditions as in Figure \ref{fig:static1}, except $l= 1$ and $\rho=5.12$, yielding $R_0\approx 2.12$ and $R_c\approx 1.82$. 
			Top row: fraction of protected individuals (solid blue), fraction of non-protected individuals (solid red), total fraction of individuals (solid yellow), and the corresponding trajectory for the model without protection (with $q=1$, dashed green), shown for the infected class (left) and the removed class (right).
			Bottom row: same as fractions of individuals within each subgroup. Left: infected fractions. Right: removed fractions. \label{fig:static3}}
	\end{center}
\end{figure}

The top row of \Cref{fig:static1} shows a simulation for \eqref{VBH} with $q = 0.2$ and $p=0.6$, which yields $R_c\approx 2.24>R_0$. Observe that, although protective behaviour leads to an increased value of the reproduction number, the final epidemic size $R^\infty_H$ is smaller than for the model without protective behaviour. This indicates that personal protection against mosquito bites can help reduce the total number of infections in the long run. Yet, the bottom row of \Cref{fig:static1} shows that, while protective behaviour decreases the proportion of individuals infected in the protected subpopulation ($\hat R_P^\infty$), the proportion infected in the non-protected subpopulation ($\hat R_{NP}^\infty$) is larger than in the model without protection. A similar phenomenon was also observed in~\cite{dye1988} in the context of a vector-borne epidemic model with heterogeneity in mosquito preferences.

In contrast, in \Cref{fig:static3}, where we take $p=0.6$ but $l=1$, protective behaviour not only effectively reduces the value of the reproduction number $R_0$, with $R_c\approx 1.82<R_0$, but also decreases $\hat R_{NP}^\infty$. Moreover, the final epidemic size $R_H^\infty$ is even slightly smaller than that observed in \Cref{fig:static1}. A possible explanation lies in the larger value of $l$ assumed here (recall that $\rho$ is adjusted accordingly to preserve $R_0\approx 2.12)$. Indeed, although the value of $R_0$ remains unchanged compared with the simulations in \Cref{fig:static1}, the larger fraction of non-competent hosts (or, under the alternative interpretation of $L$, a greater spatial separation between human hosts) acts as a ``shield'', mitigating the tendency of protective behaviour to concentrate mosquito bites in the non-protected group. 

\section{Information-dependent behavioural changes}\label{sec:model+info}
In this section, we modify model \eqref{VBH} to allow individuals to dynamically adjust their behaviour based on the information currently available about an ongoing epidemic. To model this scenario, we follow the approach of~\cite{dOnofrioManfredi2009, d2022behavioral, d2007vaccinating} by introducing the \emph{information index} $J(t)$, 
which provides a summary of the publicly available information on the infection at time $t\ge 0$.
In particular, we assume that $J(t)$ depends on both the present and past prevalence~\citep{dOnofrioManfredi2009} of the vector-borne disease in the human population. More precisely, we define $J\colon [0, \infty)\to \R_{\ge 0}$ as

\begin{align}\label{infindex}
	J(t) :&=\int_{-\infty}^t (I_{P}(\theta) + I_{NP}(\theta))K(t-\theta)\dd \theta,
\end{align}
where 
$K$ is a non-negative measurable memory kernel satisfying
$\int_0^{\infty}K(\tau)\dd \tau = 1$, which weights the contribution of current and past infections to the information available at present. 
We note that, in principle, $J(t)$ may depend on a variety of (epidemiological) variables. For instance, it may depend on the current and past values of the human incidence, as in~\cite{d2022behavioral}; see also~\cite{buonomo2025integral}. Moreover, present and past prevalence values might be translated into information through a nonlinear increasing function $g$.
For the sake of simplicity, however, we restrict our attention to the prevalence-dependent and linear case.

We consider model \eqref{VBH} again. We assume that individuals can adjust their behaviour based on information about present and past prevalence, that is, on $J(t)$. This implies that individuals may move between group $H_P$ and group $H_{NP}$, according to the following system of ODEs:
\begin{equation}\label{modelbehH}
	\left\{\setlength\arraycolsep{0.1em}
	\begin{array}{rl} 
		H_{P}' &= a(J)H_{NP}  - w(J) H_{P}  ,\\[3mm]
		H_{NP}' &=  -a(J)H_{NP}  + w(J) H_{P},
	\end{array} 
	\right.
\end{equation}
which is obtained from the model proposed in~\cite{demers2018dynamic} by allowing the rates 
$a$ and $w$ 
to depend on $J$. 
We interpret variations in $a$ and $w$ as resulting from public health campaigns encouraging individuals to adopt protective measures against mosquito bites, as well as from the willingness of host individuals to comply with such recommendations.
In particular, we assume that $a$ and $w$ satisfy the following requirements:
\begin{assumption}\label{ass:aw} 
	We assume that $a$ is a positive and increasing function, while $w$ is a positive and non-increasing function; that is, $a(x),\, w(x)>0$, $a'(x)>0$ and $w'(x)\le 0$ for all $x\in [0, +\infty)$. 
\end{assumption}
\Cref{ass:aw} models a scenario in which public health campaigns and increasing information about the current epidemic have a beneficial effect on the host population, enhancing its awareness of the disease and, consequently, its willingness to adopt self-protective measures. Such assumptions are reasonable in the case of a first large outbreak of an infectious disease for which the host population has no prior experience with past epidemics, as in the case of vector-borne diseases in temperate European countries such as Italy.

Figure \ref{Hdyn} provides a schematic representation of the flow of \eqref{modelbehH}.
\begin{figure}
	\centering
	\begin{tikzpicture}[draw, roundnode/.style={circle, draw=black, thin,
			minimum size=1.2cm},
		align=center, node distance = 1cm and 2cm, > = Stealth, accepting/.style = accepting by arrow,   >=latex,shorten >=1pt, shorten <=1pt]
		\node[roundnode] (Hp) {$H_P$};
		\node[roundnode,  right = of Hp%
		] (Hnp) {$H_{NP}$};
		\path[->] 
		(Hp) edge[bend left= 35]              node[bend right, anchor=center, above=0.3em]            {$w(J)$} (Hnp);
		\path[->] 
		(Hnp) edge[bend left= 35]              node[bend right, anchor=center, below=0.3em]            {$a(J)$} (Hp);
		
	\end{tikzpicture}
	\caption{Flow chart for model \eqref{modelbehH}.\label{Hdyn}}
\end{figure}
Let $p(t) \coloneqq H_P(t)/H$ denote the fraction of individuals in the population adopting protective behaviour at time $t$ (so that the fraction of those not adopting it is $1 - p$ at all times). Using $(H_P + H_{NP})' \equiv 0$, it follows that \eqref{modelbehH} can be conveniently rewritten as
\begin{equation}\label{eqdiffp}
	p'=a(J)(1-p)-w(J)p
	=a(J)-[a(J)+w(J)]p.    
\end{equation}
Then, taking into account \eqref{eqdiffp}, we modify system~\eqref{VBH} as follows: 
\begin{equation}\label{VBsystemAP}
	\left\{\setlength\arraycolsep{0.1em}
	\begin{array}{rl} 
		S_{P}' &= -\beta_{H\leftarrow M} \rho I_M \cfrac{q S_{P} }{c(p,q)+l} + a(J)S_{NP}  - w(J) S_{P}  ,\\[3mm]
		S_{NP}' &=  -\beta_{H\leftarrow M} \rho I_M \cfrac{S_{NP} }{c(p,q)+l}-a(J)S_{NP}  + w(J) S_{P}  ,\\[3mm]
		I_{P}' &= \beta_{H\leftarrow M} \rho I_M \cfrac{q S_{P} }{c(p,q)+l} - \gamma  I_{P}  + a(J)I_{NP}  - w(J) I_{P}  ,\\[3mm]
		I_{NP}' &= \beta_{H\leftarrow M} \rho I_M \cfrac{S_{NP} }{c(p,q)+l} - \gamma  I_{NP}  - a(J)I_{NP}  + w(J)  I_{P}  ,\\[5mm]
		p' &= a(J)-\left[a(J)+w(J)\right]p,\\[3mm] 
		I_M' &= \beta_{M\leftarrow H}(1-I_M) \cfrac{q I_{P} +I_{NP} }{c(p,q)+l} -\mu I_M ,\\[3mm]
		J(t) &=\displaystyle\int_{-\infty}^t (I_{P}(\theta) + I_{NP}(\theta))K(t-\theta)\dd \theta.
	\end{array} 
	\right.
\end{equation}
where $c(p,q)=1-p(1-q)$, recalling \eqref{c(p,q)}. 

\begin{remark}\label{remark_roosa}
	A model similar to \eqref{VBsystemAP} was proposed in~\cite{roosa2022}, where the authors considered a much more complex system (including, among other features, human demography, stages of the mosquito life cycle, and additional interventions targeting the mosquito population) and allowed only susceptible individuals to change their behaviour with rates dependent on the current prevalence (i.e. in our notation, $J\coloneqq I_H$). In contrast, protective behaviour in the infected and removed classes was assumed to be static. The model proposed here extends that of~\cite{roosa2022} by allowing all human individuals to adjust their behaviour according to possibly delayed information on the prevalence.
\end{remark}
In the following, for the sake of simplicity, 
we assume, as is common in the literature~\cite{d2022behavioral, d2007vaccinating}, that the memory kernel is given by
\begin{equation}\label{erlang}
	K(\theta)=\cfrac{k^{n}\theta^{n-1}e^{-\kappa \theta}}{(n-1)!},\qquad \theta\in \R_{\ge0},\quad n\in \N,\quad k\in (0,+\infty),
\end{equation}
that is, $K$ is the density function of an Erlang distribution with shape parameter $n$, mean $\varphi = n/k$ (where $k$ is the rate parameter), and standard deviation $\sigma = \sqrt{n}/k$.
This ``simplifying'' assumption on $K$ allows us to use the Linear Chain Trick~\citep{macdonald_time_1978, macdonald1989} to reduce the equation for $J$ to a system of ODEs \citep{diekmann2018finite}. In particular, for $n=1$, one can rewrite the equation for $J$ as
\begin{align}\label{expM}
	J' =k\left(I_{P}  + I_{NP} \right)-k J,
\end{align}
while for $n\ge 1$, it can be reduced to the following system of ODEs:
\begin{equation}\label{ErlangM}
	\left\{\setlength\arraycolsep{0.1em}
	\begin{array}{rl} 
		Z_1' &=k\left(I_{P}  + I_{NP} \right)-k Z_1, \\[2mm]
		Z_i' &=kZ_{i-1} -kZ_i ,\qquad \qquad \qquad i=2,\dots, n,
	\end{array} 
	\right.
\end{equation}
with $Z_n =J$.

In the following, for the sake of generality, we avoid using \eqref{expM} and refer only to \eqref{ErlangM}, with the convention that, for $n=1$, system \eqref{ErlangM} implicitly coincides with \eqref{expM} and $Z_1=Z_{n-1}=J$.
Moreover, we refer to $\varphi$ as the ``delay'', as it qualitatively represents the lag between the current time and the peak of the memory kernel.
Note that, as $k\to +\infty$ and $n$ is bounded, one has $\varphi, \sigma \to 0^+$; that is, $K$ converges in distribution to a Dirac delta concentrated at $0$. In this case, the information becomes instantaneous.
On the other hand, if both $n$ and $k$ increase in such a way that $\varphi$ remains constant, then the Erlang distribution becomes increasingly concentrated around its mean $\varphi$.  
In particular, if $n,\, k\to +\infty$ with $n/k=\varphi$, then the $n$-Erlang distribution approaches a discrete delay at a linear rate~\citep{ando2020}, i.e. it converges to a Dirac delta concentrated at $\varphi$. In this case, the model reduces to one with memory concentrated at $t-\varphi$.

For later use, we define
\begin{equation*}
	S_H\coloneqq S_P+S_{NP},\quad\text{and}\quad I_H\coloneqq I_P+I_{NP}.  
\end{equation*}
Since we are dealing with fractions, the biologically meaningful state space for the solutions of model \eqref{VBsystemAP} turns out to be the set
\begin{equation*}
	\begin{split}
		\Omega:=\{ (&S_P, I_P, S_{NP}, I_{NP}, p, I_M, Z_1, \dots, Z_{n-1}, J) \in \R^{6+n}_{\ge 0}\ :\\
		\label{Omega_static_2} &0\le S_P, I_P, S_{NP}, I_{NP}, p, I_M, Z_i, J\le 1,\quad i=1,\dots, n-1 ,\ S_H+I_H\le 1\},    
	\end{split}
\end{equation*}
for which the following result holds.
\medskip
\begin{proposition}\label{prop_ex_uniq2}
	For every initial condition
	$$X(0)\coloneqq (S_P(0), I_P(0), S_{NP}(0), I_{NP}(0), p(0), I_M(0), Z_1(0),\dots, Z_{n-1}(0), J(0)) \in \Omega,$$ system \eqref{VBsystemAP} admits a unique solution $$X(t)\coloneqq (S_P(t), I_P(t), S_{NP}(t), I_{NP}(t), p(t), I_M(t), Z_1(t),\dots Z_{n-1}(t), J(t)) \in \Omega,$$ which is globally defined in the future. 
	Furthermore, the inequality
	$S_H(t)+I_H(t)\leq  S_H(0)+I_H(0)$
	holds for every $t\ge 0$. 
\end{proposition} 
\begin{proof}
	Let $X(0)\in \Omega$. The local existence and uniqueness of the solution $X(t)$ of the Cauchy problem relevant to \eqref{VBsystemAP} follows from the Cauchy--Lipschitz theorem. 
	Now, for each $$Z\in\{ S_P, I_P, S_{NP}, I_{NP}, p, I_M, Z_1, \dots, Z_{n-1}, J\},$$ one has
	$Z'|_{Z=0}\geq 0$, which implies that $X(t)\in \R^{6+n}_{\ge 0}$ for all $t\ge 0$ for which the solution is defined. Furthermore, the inequalities 
	$(S_H+I_H)'=-\gamma I_H\leq0$, 
$p|_{p\ge 1}'\le -\omega(J) p\le 0$ and $I'_M|_{I_M\ge 1} \le -\mu I_M\leq 0$ imply that $S_P, I_P, S_{NP}, I_{NP}, p, I_M\le 1$. 
From this, one also gets $Z_i|_{Z_i\ge 1}\le 0$ for $i=1,\dots, n$ (recall $Z_n=J$), which gives $Z_i\le 1$. Hence, standard arguments give that $X(t)\in \Omega$ is globally defined in the future and $S_H(t)+I_H(t)\leq  S_H(0)+I_H(0)$ for all $t\ge 0$.
\end{proof}

\subsection{Control reproduction number}\label{R0DynamicChanges}
As for the model with static behaviour, in this section we compute the CRN $\hat R_c$ for model \eqref{VBsystemAP}. 
For any $a_0\coloneqq a(0)>0$ and $w_0\coloneqq w(0)>0$, the model \eqref{VBsystemAP} admits the DFE
\begin{equation}\label{DFEdyn}
\left(S_{P}^*, S_{NP}^*, I_{P}^*, I_{NP}^*, p^*, I_M^*, J^*\right) \coloneqq \left(p_0, 1-p_0, 0, 0, p_0, 0, 0\right),
\end{equation}
where
\begin{equation}\label{eqp}
p_0\coloneqq \cfrac{a_0}{a_0+w_0}.
\end{equation}
Linearising the equations for the infected individuals around the DFE~\eqref{DFEdyn}, we obtain
\begin{equation*}
\left\{\setlength\arraycolsep{0.1em}
\begin{array}{rl} 
	I_{P}' &= \beta_{H\leftarrow M}\rho I_M \cfrac{p_0q}{c(p_0,q)+l} - \gamma I_{P} + a_0I_{NP} - w_0 I_{P}  ,\\[3mm]
	I_{NP}' &= \beta_{H\leftarrow M} \rho I_M \cfrac{1-p_0}{c(p_0,q)+l} - \gamma  I_{NP}  - a_0I_{NP}  + w_0  I_{P}  ,\\[3mm]
	I_M' &= \beta_{M\leftarrow H} \cfrac{q I_{P} +I_{NP} }{c(p_0,q)+l} -\mu I_M.
\end{array} 
\right.
\end{equation*}
Then, considering an infection matrix $B$ defined as in \eqref{infection:matrix} and the transition matrix
\begin{equation}\label{tranMat}
\hat \Sigma\coloneqq 
\begin{pmatrix} 
	\gamma + w_0  & -a_0 & 0\\
	-w_0 & \gamma + a_0 & 0\\
	0 & 0 & \mu 
\end{pmatrix},
\end{equation}
the CRN $\hat R_c=\hat {R}_c(p_0, q)$ is obtained as the spectral radius of $\hat K\coloneqq B\hat \Sigma^{-1}$, and, for $p_0$ as in~\eqref{eqp}, it reads
\begin{align}
\label{hatR0beh}\hat R_c(p_0, q)
= R_0\, \frac{1+l}{c(p_0,q) + l}\, \sqrt{
	\Delta
},
\end{align}
where
$$
\Delta = \frac{qp_0\left(a_0 q+\gamma q+w_0\right)+(1-p_0)\left(a_0q+w_0+\gamma\right)}{a_0+w_0+\gamma}.
$$
Note that 
$\Delta$
is non-negative for $a_0,\, w_0\in (0,+\infty)$; hence, $\hat R_c(p_0, q)\ge 0$ is well-defined. 
In particular, the following result holds.
\begin{proposition}\label{demers}
Let $R_c$ be defined as in \eqref{R0het} and $a_0,\, w_0\in (0, +\infty)$. Then 
\begin{equation}\label{RcminR0}
	\hat R_c(p_0, q)<R_c(p_0, q),\quad \text{for all}\quad p_0\in(0, 1],\ q\in [0,1),
\end{equation}
Moreover, if $a_0,\, w_0\to 0^+$ in such a way that $\underline{p}\coloneqq \lim_{a_0, w_0\to 0^+} p_0$ exists finite, then  
\begin{equation}\label{lim_to_stat}
	\lim_{a_0,\, w_0\to 0^+}\hat R_c(p_0, q)= R_c(\underline{p}, q).
\end{equation}
If instead $a_0,w_0\to +\infty$ in such a way that $\bar{p}\coloneqq \lim_{a_0, w_0\to +\infty} p_0$ exists finite, then 
\begin{equation}\label{lim_to_1_group}
	\lim_{a_0, w_0\to +\infty} \hat R_c(p_0, q) = R_0\, \cfrac{c(\bar p,q)(1+l)}{c(\bar p,q)+l}<R_0,\quad \text{for all}\quad\bar p\in (0, 1],\ q\in [0, 1)\,.
\end{equation}
\end{proposition}
\begin{proof}
Observe that
\begin{align}
	\label{firstR0beh}
	\Delta
	=1-
	p_0(1-q)-\frac{(1-q)(a_0 c(p_0,q)+\gamma p_0 q)}{a_0+w_0+\gamma}.
\end{align}
Hence, relation \eqref{RcminR0} immediately follows from \eqref{hatR0beh} and \eqref{firstR0beh}, since \eqref{firstR0beh} is smaller than 1 for $a_0,\, w_0\in (0, +\infty)$ and $q\in[0, 1)$. As for the limits in \eqref{lim_to_stat} and \eqref{lim_to_1_group}, it suffices to observe that
$$\lim\limits_{a_0,w_0\to 0^+}\frac{(1-q)(a_0 c(p_0,q)+\gamma p_0 q)}{a_0+w_0+\gamma}=(1-q)p_0q,$$
and
\begin{equation*}
	\lim_{a_0,w_0\to +\infty}\frac{p_0(1-q)^2w_0}{a_0+w_0+\gamma}=\lim_{a_0,w_0\to +\infty}\frac{(1-q)(a_0 c(p_0,q)+\gamma p_0 q)}{a_0+w_0+\gamma}=c(\bar p,q)(1-q)\bar p,
\end{equation*}
from which the claim follows.
\end{proof}

As an application of \Cref{demers}, let us consider epidemiological parameters as in \Cref{TableValues}, with $l=0.25$ and $\rho=2$, $p=0.6$ and $q=0.2$ as in \Cref{fig:static1}. Then formula \eqref{R0het} gives $R_c\approx 2.24>2.12\approx R_0$.
Let us take $w_0=0.1,\ 1,\ 10,\ 100$ and let $a_0=w_0p/(1-p)$, 
so that $p_0\coloneqq p$.
Then \eqref{hatR0beh} gives $\hat R_c\approx 2.00$ for $w_0=0.1$, $\hat R_c\approx 1.83$ for $w_0=1$, $\hat R_c\approx 1.79$ for $w_0=10,\ 100$, while the limit in \eqref{lim_to_1_group} gives $\hat R_c\approx 1.79$. Note that for all these choices of $w_0$, one has $\hat R_c< R_0<R_c$.

\section{Separation of time scales}\label{sec:GSPT1}

In this section, we assume that the dynamics of model \eqref{VBsystemAP} evolves on two distinct time scales. In particular, as is common in the literature, we assume that behavioural changes occur on a faster time scale than the epidemiological dynamics. On the one hand, as observed in~\cite{demers2018dynamic}, if the epidemiological dynamics evolves much faster than the behavioural ones, then model~\eqref{VBsystemAP} effectively reduces, on the fast time scale, to the model with static protective behaviour in \eqref{VBH}. In this case, the dynamics induced by protective behaviour is entirely determined by the initial fraction of individuals adopting protection, with no information-driven effects.
On the other hand, it is more interesting, from both a biological and a mathematical perspective, to investigate the role of information-induced behavioural changes when these occur on a much faster time scale than the epidemiological dynamics. This assumption is also common in the literature; see, for instance,~\cite{bulai2024geometric, della2024geometric, poletti2009spontaneous}. We also note that in~\cite{demers2018dynamic}, the authors heuristically discussed the case of infinitely large (and information-independent) switching rates, relating a two-group model with dynamical behavioural changes to a homogeneous-population model.
In contrast to the previous discussion, in this section, we provide a rigorous derivation of this reduction to a homogeneous population model using an approach based on GSPT. 

To this aim, we assume that $\beta_{H\leftarrow M},\,\beta_{M\leftarrow H},\,\gamma,\,\mu \in \mathcal{O}(\varepsilon)$, for some $0<\varepsilon\ll 1$, while all remaining parameters are $\mathcal{O}(1)$. This qualitatively corresponds to assuming that opinion spreads much faster than the disease and the relevant information, i.e.~the system evolves on two distinct time scales. Under these assumptions, with a slight abuse of notation to avoid the introduction of four new parameters, model \eqref{VBsystemAP} may be rewritten as
\begin{equation}\label{epsilonmodel}
\left\{\setlength\arraycolsep{0.1em}
\begin{array}{rl} 
	S_{P}' &= -\varepsilon\beta_{H\leftarrow M}\rho I_M \cfrac{q S_{P} }{c(p,q)+l} + a(J)S_{NP}  - w(J) S_{P}  ,\\[3mm]
	S_{NP}' &=  -\varepsilon\beta_{H\leftarrow M}\rho I_M \cfrac{S_{NP} }{c(p,q)+l}-a(J)S_{NP}  + w(J) S_{P}  ,\\[3mm]
	I_{P}' &= \varepsilon\beta_{H\leftarrow M} \rho I_M \cfrac{q S_{P} }{c(p,q)+l} - \varepsilon\gamma  I_{P}  + a(J)I_{NP}  - w(J) I_{P}  ,\\[3mm]
	I_{NP}' &= \varepsilon\beta_{H\leftarrow M} \rho I_M \cfrac{S_{NP} }{c(p,q)+l} - \varepsilon\gamma  I_{NP}  - a(J)I_{NP}  + w(J) I_{P}  ,\\[6mm]
	p' &= a(J)-\left[a(J)+w(J)\right]p,\\[3mm] 
	I_M' &= \varepsilon\beta_{M\leftarrow H}(1-I_M) \cfrac{q I_{P} +I_{NP} }{c(p,q)h+l} -\varepsilon\mu I_M ,\\[5mm]
	Z_1' &=\varepsilon k\left(I_{P}  + I_{NP} \right)-\varepsilon k Z_1 \\[5mm]
	Z_i' &=\varepsilon kZ_{i-1} -\varepsilon kZ_i , \qquad i=2,\dots, n,
\end{array} 
\right.
\end{equation}
with $Z_n =J$. Letting $\varepsilon\to 0^+$ in system \eqref{epsilonmodel}, we obtain the corresponding so-called \emph{layer system}
\begin{equation}\label{fastvar}
\left\{\setlength\arraycolsep{0.1em}
\begin{array}{rl} 
	X_{P}' &= a(J)X_{NP}  - w(J) X_{P}  ,\\[3mm]
	X_{NP}' &= -a(J)X_{NP}  + w(J) X_{P},\quad \text{for}\quad X\in \{S, I\},\\[3mm]
	p' &= a(J)-\left[a(J)+w(J)\right]p,\\[3mm]
	J'&=0.
\end{array} 
\right.
\end{equation}
Note that we omit the equations for $I_M$ and $Z_1,\dots, Z_{n-1}$, since they neither evolve on this time scale, nor appear in the remaining equations.
Moreover, we observe that 
\begin{equation}\label{def_X_H}
X_H\coloneqq X_{P}+X_{NP},\quad \text{for}\quad X\in \{S, I\},
\end{equation}
satisfies $X_H'=0$. Therefore, the layer system~\eqref{fastvar} can be rewritten as 
\begin{equation}\label{fastvarreduced}
\left\{\setlength\arraycolsep{0.1em}
\begin{array}{rl} 
	S_{P}' &= a(J)S_H  -[a(J)+ w(J)] S_{P}  ,\\[3mm]
	I_{P}' &= a(J)I_H  -[a(J)+ w(J)] I_{P} ,\\[3mm]
	p' &= a(J)-\left[a(J)+w(J)\right]p,\\[3mm]
	S_H'&=I_H'=J'=0.
\end{array} 
\right.
\end{equation}
Let $\tilde\Omega:=[0,1]^6$. The set of equilibria of system \eqref{fastvarreduced} in $\tilde\Omega$  is given by
\begin{align}\label{eq:crit_man_3}
\mathcal{C}_0\coloneqq \left\{ (S_{P},I_{P},p, S_H, I_H, J)\in \tilde \Omega \ \bigg| \  S_{P}=pS_H, \ I_{P}=pS_H,\ p=\frac{a(J)}{a(J)+w(J)}\right\}.
\end{align}
In GSPT, this set is known as the \emph{critical manifold} of system \eqref{epsilonmodel}. The following result holds.

\begin{proposition}\label{propAT}
The set $\mathcal C_0$ defined in \eqref{eq:crit_man_3} is globally attracting for system \eqref{fastvarreduced} on $\tilde\Omega$.
\end{proposition}
\begin{proof}
For any $(S_{P}(0),I_{P}(0), p(0),S_H(0), I_H(0), J(0))\in \tilde \Omega$, the solutions of system~\eqref{fastvarreduced} are explicitly given by
\begin{equation*}
	X_{P}(t) = e^{-[a(J)+w(J)] t}X_{P}(0)+ \cfrac{a(J)}{a(J)+w(J)}X_H(1-e^{-[a(J)+w(J)]t}),
\end{equation*}
for $(X_P, X_H)\in \{(S_P, S_H),\  (I_P, I_H),\ (p,1)\}$, which gives $(S_{P}(t), I_{P}(t), p(t))\to \mathcal{C}_0$ as $t\to +\infty$.
\end{proof}

\Cref{propAT} ensures that the dynamics of system \eqref{epsilonmodel} is attracted to the critical manifold \eqref{eq:crit_man_3}, where the small parameters begin to play a major role. 
Now we study the slow dynamics of the system by considering the slow time scale $\tau =\varepsilon t$.
In this setting, we have 
\begin{equation*}
S_{P}=p(J)S_H,\quad I_{P}=p(J)I_H,\quad     p(J)=\cfrac{a(J)}{a(J)+w(J)}. 
\end{equation*}
Thus, system \eqref{epsilonmodel} reduces to the \emph{one-group model}
\begin{equation}\label{slowINFsystem}
\left\{\setlength\arraycolsep{0.1em}
\begin{array}{rl} 
	\dot{S}_H &= -\beta_{H\leftarrow M}\rho I_M h(p(J), q) S_H,\\[3mm]
	\dot{I}_H &= \beta_{H\leftarrow M} \rho I_M h(p(J), q) S_H - \gamma  I_H  ,\\[3mm]
	\dot{I}_M &= \beta_{M\leftarrow H}(1-I_M) h(p(J), q) I_H  -\mu I_M,\\[3mm]
	\dot Z_1 &=k I_H- k Z_1 \\[3mm]
	\dot Z_i &= kZ_{i-1} - kZ_i , \qquad i=2,\dots, n,
\end{array} 
\right.
\end{equation}
with $Z_n =J$, where $\dot Y \coloneqq \frac{\text{d}}{\text{d}\tau}Y$, and \begin{equation}
h(p,q)\coloneqq \frac{c(p, q)}{c(p,q)+l},\qquad p\in[0, 1],\quad q\in [0, 1).
\label{def_H}
\end{equation}
Note that $h(p, q)\ge 0$ for $p\in[0, 1]$ and $q\in [0, 1)$. Moreover, since \Cref{ass:aw} ensures that $w(x)>0$ for $x\geq 0$, we obtain
\begin{equation*}
h(p(x), q)=\frac{w(x)+qa(x)}{[w(x)+qa(x)]+l[w(x)+a(x)]}=\frac{1}{1+l\frac{w(x)+a(x)}{w(x)+qa(x)}}\,.  
\end{equation*}
\Cref{ass:aw} also ensures that 
\begin{equation*}
p'(x)=\cfrac{a'(x)w(x)-a(x)w'(x)}{[a(x)+w(x)]^2}>0,\qquad x\in [0, +\infty),   
\end{equation*}
thus
\begin{equation}\label{der:h}
\partial_p h(p(x), q)\cdot  p'(x)=-\cfrac{l(1-q)p'(x)}{[1-(1-q)p(x)+l]^2}<0,\qquad x\in [0,+\infty).
\end{equation}

\begin{remark}
We have $h(p(x), q)>0$ for $q\in [0, 1)$ and $x\in [0, +\infty)$. 
Thus, equation~\eqref{der:h} ensures that $h(p(x), q)$ admits limit for $x\to+\infty$, and the largest reduction for the transmission rates $\beta_{H\leftarrow M},\,\beta_{M\leftarrow H}$ is obtained at $\bar h\coloneqq \lim_{x\to+\infty} h(p(x), q)$.
\end{remark}

Lastly, we establish the following result on the asymptotic behaviour of the solutions of system~\eqref{slowINFsystem}, which follows directly by arguing as in the proof of \Cref{prop:extinction} and using \Cref{prop_ex_uniq2}.
\begin{proposition}\label{extinction_1group}
For system \eqref{slowINFsystem}, one has $\lim_{\tau\to+\infty} I_H(\tau) =\lim_{\tau\to+\infty} I_M(\tau) =0$.
\end{proposition}
Consequently, as expected, the dynamics of \eqref{slowINFsystem} always ends with epidemic extinction, independently of the hosts' protective behaviour. However, we are interested in the transient behaviour of the solutions of system~\eqref{slowINFsystem}, which is far from trivial due to the absence of non-trivial equilibria (i.e. with $I_H\ne 0$ or $I_M\ne 0$).

\subsection{Numerical simulations}\label{sec:numerics}

In this section, we present numerical experiments for the proposed model with behavioural changes, comparing the outputs of model \eqref{VBsystemAP} with those of the \emph{fast} system \eqref{fastvarreduced}.
For the simulations, we use epidemiological parameters appropriate for a dengue outbreak, listed in \Cref{TableValues}.
In addition, we consider $l=0.25$ and $l=1$, and we vary $\rho$ to preserve the value of $R_0\approx 2.12$, recalling \eqref{R0}. More precisely, we take $\rho=2$ for $l=0.25$ and $\rho= 5.12$ for $l=1$.

We assume that, at the beginning of the outbreak, the human population is at the opinion equilibrium, i.e. $p(0)=p_0$, for $p_0$ defined as in \eqref{eqp}.
As for the rates $a$ and $w$, we consider the following forms:
\begin{equation*}
a(x)=a_0+\alpha x,\quad a_0,\ \alpha\ge0,\quad\text{and}\quad     w(x)\equiv w_0,\quad w_0>0.
\end{equation*}
We then vary $w_0>0$, while $a_0=w_0$, so that, according to \eqref{eqp}, $p_0=0.5$, i.e.\ at the beginning of the epidemic, $50\%$ of the host population adopts personal protective measures against mosquito bites, possibly an optimistic assumption. We set $\alpha =  \chi\, w_0$, where $\chi>0$ accounts for both the effect of PHS campaigns and the human reaction to information.
Finally, we vary $q\in [0, 1)$, the shape parameter $n\in \N$, and the delay $\varphi$ of the Erlang distribution in \eqref{erlang}. Since we observe the outbreak from its very beginning, we consider the initial conditions  \begin{gather*}
S_P(0)=p_0,\quad S_{NP}(0)=1-S_P(0),\quad I_P(0)=I_{NP}(0)=0,\quad I_M(0)=10^{-4}.
\end{gather*}
The numerical results are obtained using the MATLAB built-in ODE solvers \texttt{ode45} and \texttt{ode23s} with standard tolerances (the latter is employed when $w_0\ge 10$). 
\begin{remark}
The term $w_0$ in $\alpha$ is used to normalise the human response to information with respect to the magnitude of the reaction term $w_0$. 
To clarify the choice of $\alpha$, let us consider the case $J=I_H$, where $I_H$ is defined as in \eqref{def_X_H}. In this case, taking $\chi=10^4$ yields, for the one-group model \eqref{fastvarreduced}, that $a(x)=2w(x)$ exactly when $x=J=I_H=10^{-4}$, since $a_0=w_0$. Hence, for instance, if $H=10^5$, then $a(x)=2w(x)$ when the total number of infected humans at $t\ge 0$ is equal to $10$. 
Moreover, for these choices of $a$ and $w$, it is straightforward to verify that
\begin{equation*}
	h(p(x), q)=\frac{1}{2}h(0, q)\iff x =\frac{a_0q(1+2l)+l(w_0-a_0)+w_0}{\alpha[l-q(1+2l)]},\quad l\ne q(1+2l).
\end{equation*}
In particular, in the case of perfect protection, $q=0$, if $l=0.25$, the ``contact rates'' between mosquitoes and humans are halved exactly when $I_H=4\cdot 10^{-4}$ (i.e. when the total number of infected individuals is $40$), while for $q=0.1$ ($90\%$ of efficacy), one has $I_H=1.2\cdot10^{-3}$.  
\end{remark}

\begin{figure}[ht]
\begin{center}
	\includegraphics[width=.8\linewidth]{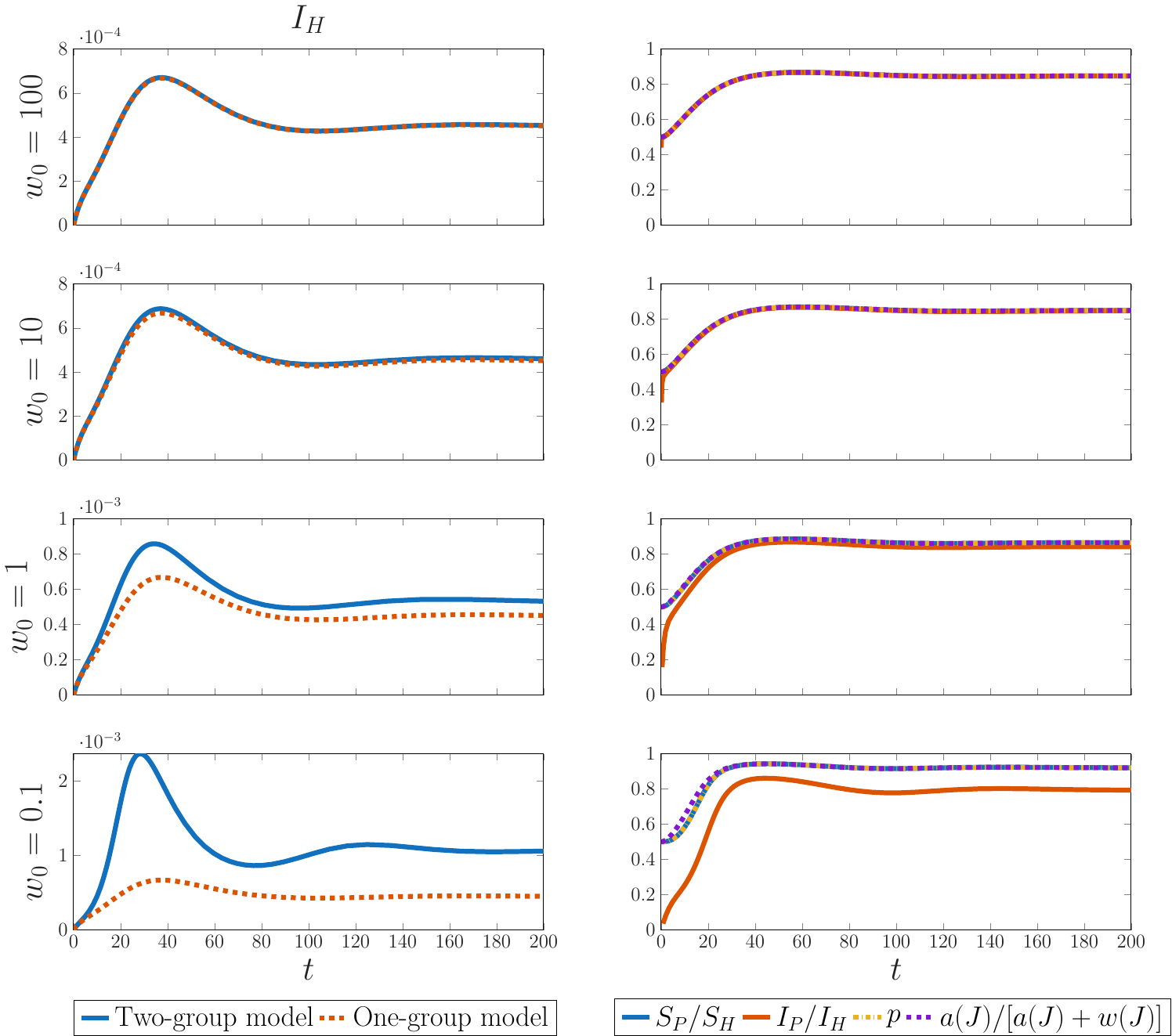}
	\caption{Comparison between \eqref{VBsystemAP} (Two-group model) and \eqref{slowINFsystem} (One-group model) with model parameters and initial conditions as in \Cref{sec:numerics} with $l=0.25$, $\rho =2$, $q=0$ (perfect protection), $\chi=10^4$, $n=1$, $\varphi=20$ days, and from bottom to top, $w_0=0.1, 1, 10, 100$ (which give $\hat R_c\approx 2.17,\  1.85,\ 1.78,\ 1.77,$ respectively). Left: $I_H$ (defined as in \eqref{def_X_H}) as a function of time for model \eqref{VBsystemAP} and \eqref{slowINFsystem}. Right: proportions $S_P/S_H$, $I_P/I_H$, $p$ and $a(J)/[a(J)+w(J)]$ as functions of time. Note that in the plots for $w_0=10, 100$, the line corresponding to $S_P/S_H$ is perfectly overlapped with those relevant to $I_P/I_H, p$ and $a(J)/[a(J)+w(J)]$, which indicates a fast convergence to the manifold $\mathcal C_0$ in \eqref{eq:crit_man_3}.
		\label{fig:2vs1group}}
\end{center}
\end{figure}

We start by comparing the outputs of \eqref{VBsystemAP} with those of \eqref{slowINFsystem} for different orders of magnitude of $w_0$. In \Cref{fig:2vs1group},
we take $l=0.25$ and $\rho=2$, and we observe that, when $w_0=10$ days$^{-1}$ or $w_0=100$ days$^{-1}$ (i.e.\ it is from two to three orders of magnitude larger than $\gamma$), the results for the two models are (almost) indistinguishable. In particular, note that in the plots for $w_0=10, 100$ in the right column, the line corresponding to $S_P/S_H$ perfectly overlaps with those corresponding to $I_P/I_H, p$ and $a(J)/[a(J)+w(J)]$, as the solution of \eqref{VBsystem4} almost immediately converges to the critical manifold $\mathcal C_0$ in \eqref{eq:crit_man_3}.  Interestingly, even for $w_0=1$, the simulation of system \eqref{slowINFsystem} provides a good approximation of model \eqref{VBsystemAP}, although the convergence to the manifold $\mathcal C_0$ is much slower in this case, leading to significantly different timing and magnitudes of the (first) epidemic peak.

Finally, we can observe that for $w_0=0.1$ (i.e.\ of the same order of magnitude as $\gamma$), the dynamics of the two models differ substantially, as the solution is still far from converging to the manifold \eqref{eq:crit_man_3}. This is also the only case, among those considered here, for which $\hat R_c\approx 2.17>2.12 \approx R_0 $; see the caption of \Cref{fig:2vs1group}. Hence, in this case, $w_0$ is not large enough to prevent protective behaviour from increasing the probability of outbreak.

Motivated by the results shown in Figure \ref{fig:2vs1group}, in the following, we consider the outputs of model \eqref{VBsystemAP} with $w_0=0.1$ and $w_0=1$, and compare them with those of \eqref{slowINFsystem}. In particular, we take $\chi=10^4$, $\varphi=30$ days, and vary the shape parameter $n$ of the Erlang distribution in \eqref{erlang}, considering $n=1,~2,~5$.

\begin{figure}
\begin{center}
	\includegraphics[width=1.\linewidth]{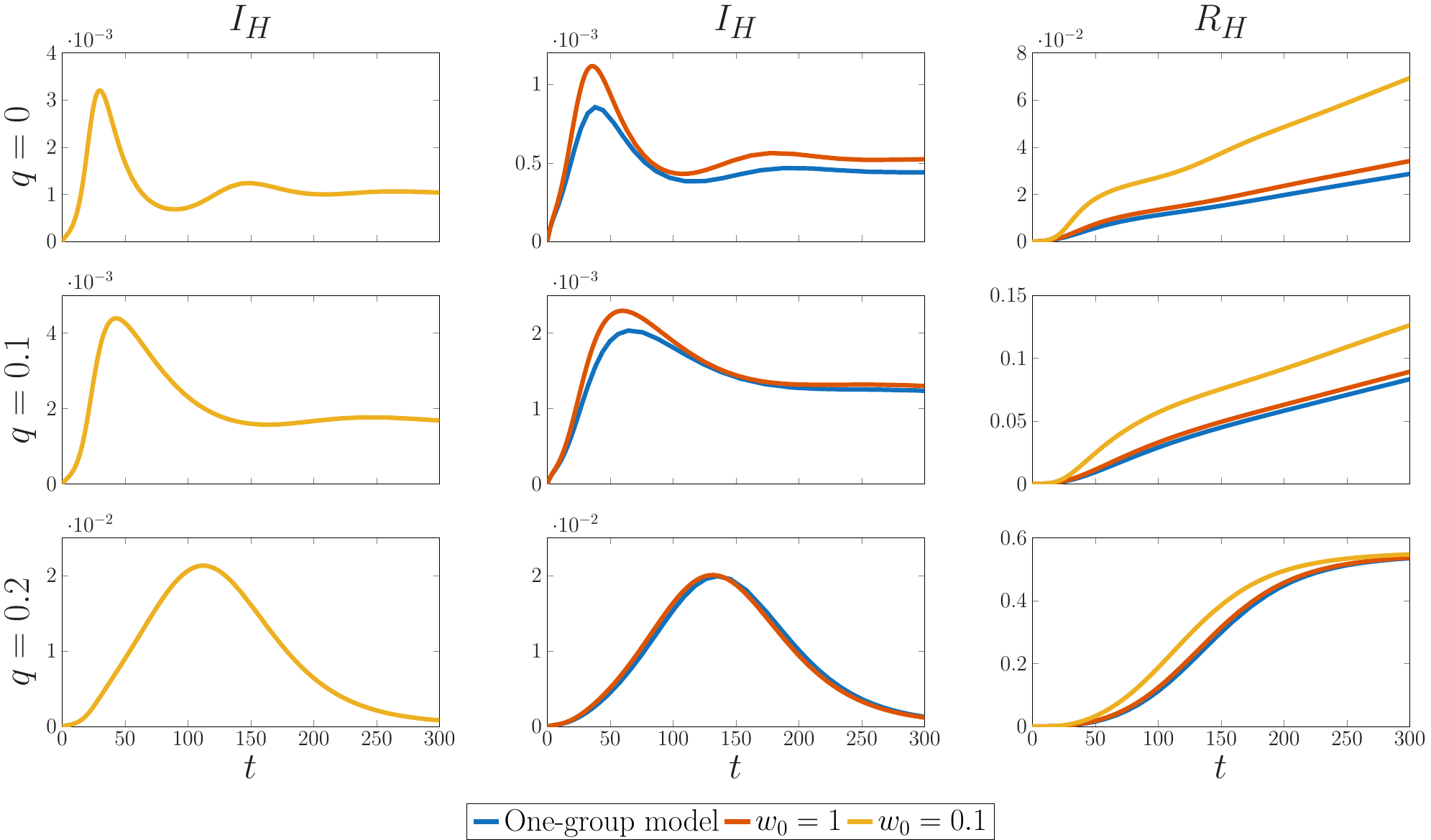}
	\caption{Plots for \eqref{VBsystemAP} (red for $w_0=1$, yellow for $w_0=0.1$) and \eqref{slowINFsystem} (blue, one-group model) with model parameters and initial conditions as in \Cref{sec:numerics}, $l=0.25$, $\rho =2$,  and $\chi =10^4$ for, from top to bottom, $q=0, 0.1,\ 0.2$, with  $n=1$ and $\varphi=30$ days in \eqref{erlang}. Left and centre: $I_H$ as a function of time. Right: $R_H$  as a function of time. \label{fig:sim_exp_phi30}}
	\vspace{3mm}
	\includegraphics[width=1.\linewidth]{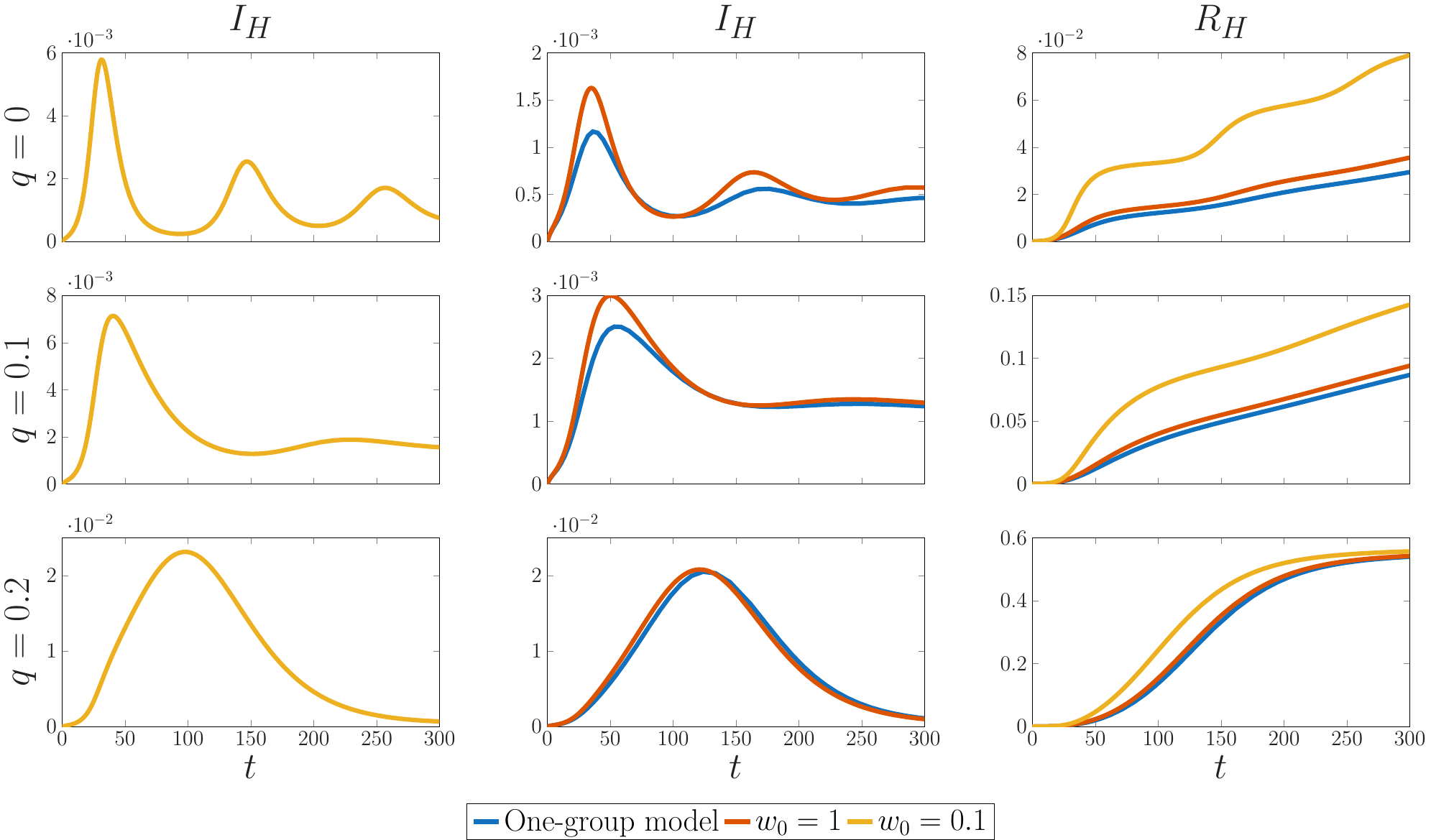}
	\caption{Plots for \eqref{VBsystemAP} (red for $w_0=1$, yellow for $w_0=0.1$) and \eqref{slowINFsystem} (blue, one-group model) with model parameters and initial conditions as in \Cref{sec:numerics}, $l=0.25$, $\rho =2$,  and $\chi =10^4$ for, from top to bottom, $q=0,\ 0.1,\ 0.2$, with  $n=2$ and $\varphi=30$ days in \eqref{erlang}. Left and centre: $I_H$ as a function of time. Right: $R_H$  as a function of time. \label{fig:sim_erl2_phi30}}
\end{center}
\end{figure}

\begin{figure}
\begin{center}
	\includegraphics[width=1.\linewidth]{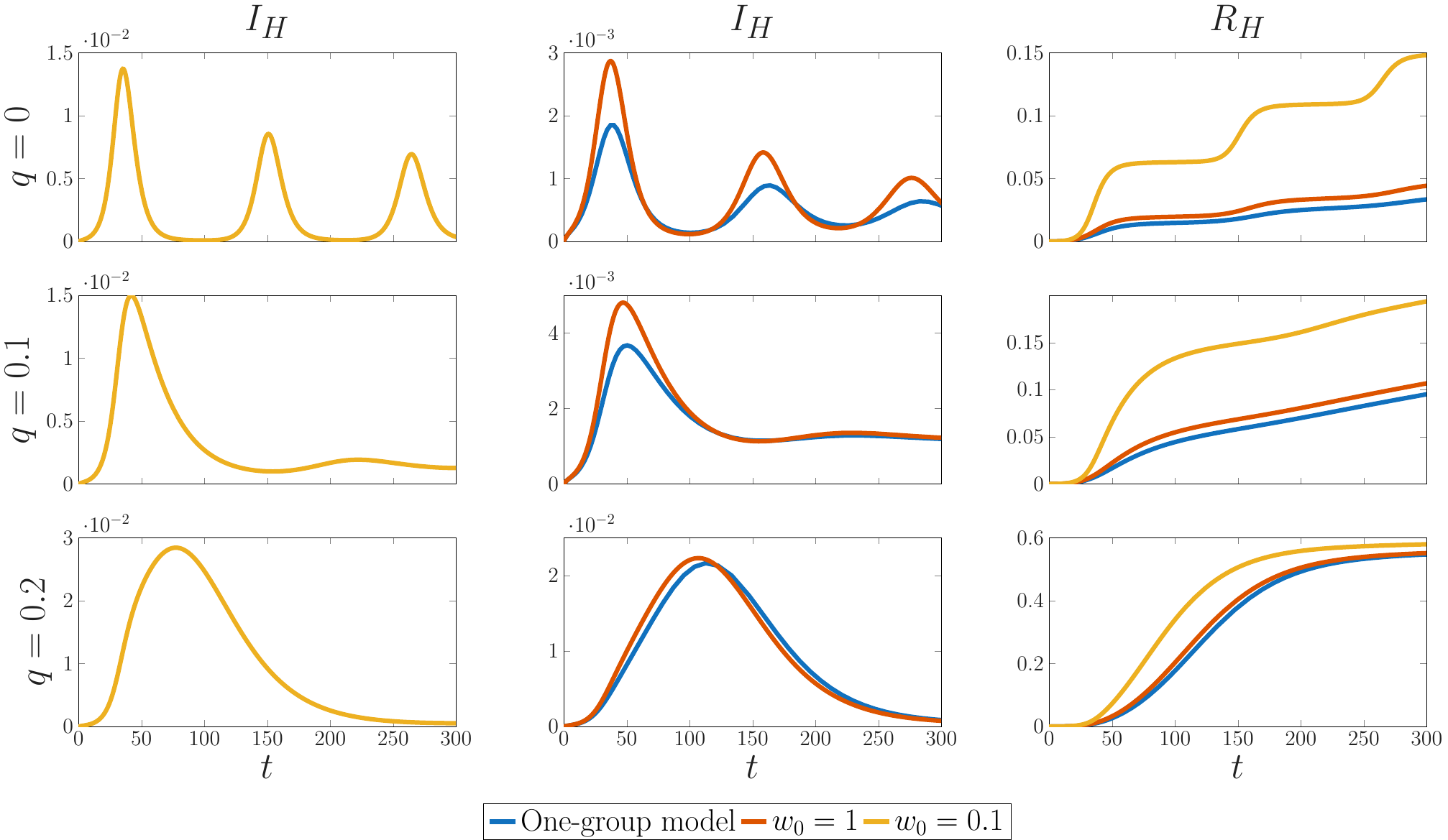}
	\caption{Plots for \eqref{VBsystemAP} (red for $w_0=1$, yellow for $w_0=0.1$) and \eqref{slowINFsystem} (blue, one-group model) with model parameters and initial conditions as in \Cref{sec:numerics}, $l=0.25$, $\rho =2$,  and $\chi =10^4$ for, from top to bottom, $q=0, 0.1,\ 0.2$, with  $n=5$ and $\varphi=30$ days in \eqref{erlang}. Left: $I_H$ as a function of time. Right: $R_H$  as a function of time. \label{fig:sim_erl5_phi30}}
	\vspace{3mm}
\end{center}
\end{figure}

\begin{figure}
\begin{center}
	\includegraphics[width=1.\linewidth]{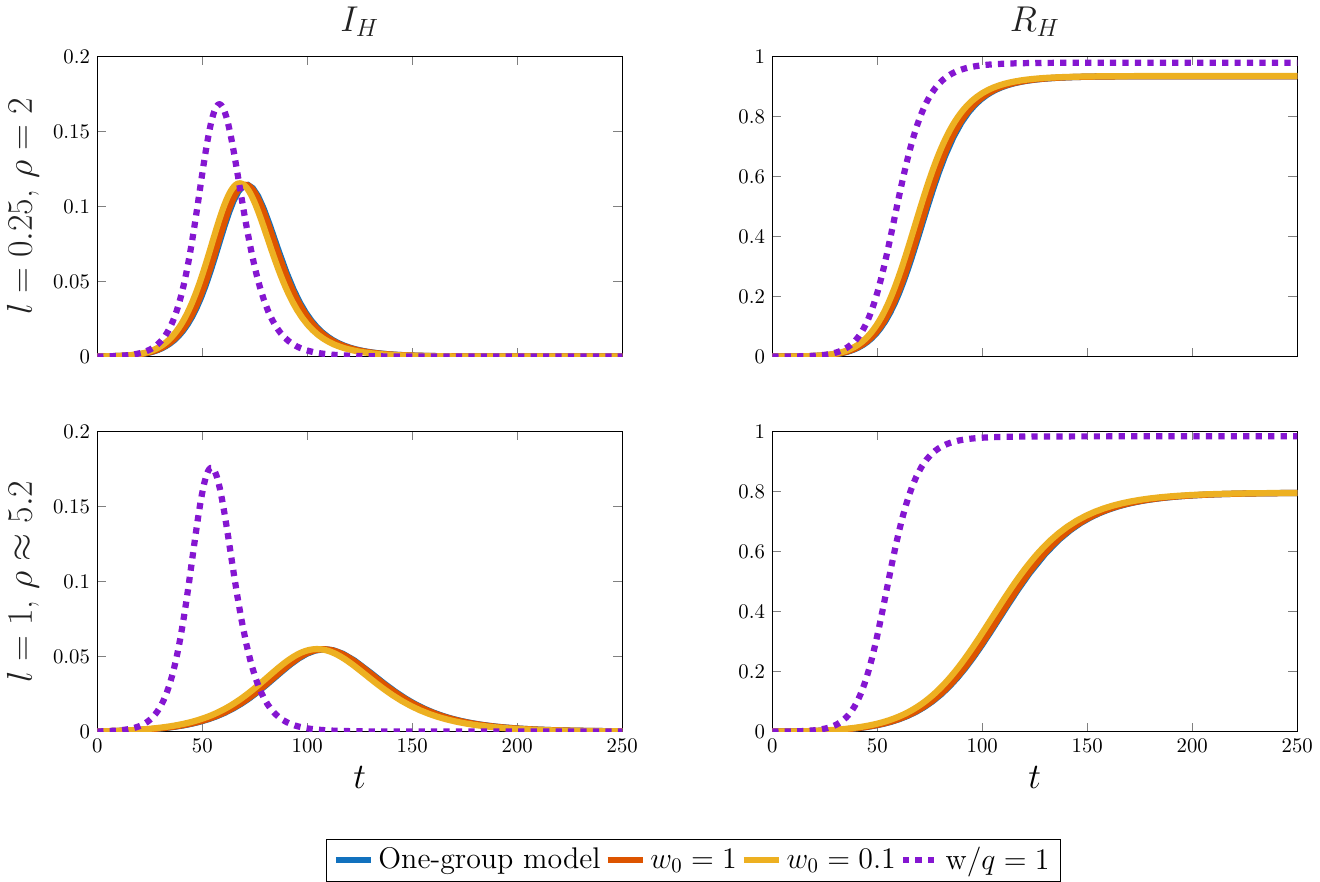}
	\caption{Plots for \eqref{VBsystemAP} (red for $w_0=1$, yellow for $w_0=0.1$) and \eqref{slowINFsystem} (blue, one-group model) with model parameters and initial conditions as in \Cref{sec:numerics}, $\chi =10^4$, $q=0.5$, $n=1$, and $\varphi=30$ days in \eqref{erlang}. Left: $I_H$ as a function of time. Right: $R_H$ as a function of time. The dashed lines represent the same simulations for the model without protective behaviour ($q=1$).  Upper row: $l=0.25$ and $\rho =2$. Lower row: $l=1$ and $\rho = 5.12$.  \label{fig:protection_05}}
\end{center}
\end{figure}

In Figures \ref{fig:sim_exp_phi30},  \ref{fig:sim_erl2_phi30},  and \ref{fig:sim_erl5_phi30}, we plot $I_H$ (left and centre) and $R_H$ (right) as functions of time for $n=1, 2, 5$, respectively, and for three different values of $q$, namely $q=0$ (upper row), $q=0.1$ (middle row), and $q=0.2$ (lower row). We observe that for $q=0$ and $q=0.1$, information-induced behavioural changes effectively prevent a large portion of infections for all choices of $n$ (note the different scales on the vertical axes for the case $q=0$ and $q=0.1$).  Interestingly, in the first and second rows, after the first epidemic peak, the solutions for $I_H$ appear to enter a quasi-stationary state, although the model does not admit any non-trivial equilibria (i.e.\ with $I_H, I_M\ne0$). Notably, this phenomenon is entirely due to the behavioural adaptation in response to the epidemic and does not correspond to convergence of 
the solutions of \eqref{VBsystemAP} and \eqref{fastvarreduced} to true equilibria of their respective systems. 
Indeed, observe that $R_H$ increases in all cases, indicating that no stationary situation has been reached by the solution, and $I_H$ will inevitably go extinct in the long run, as shown in \Cref{extinction_1group}. Furthermore, when the memory distribution is more concentrated around its mean (and hence further in the past), as in Figures \ref{fig:sim_erl2_phi30} and \ref{fig:sim_erl5_phi30}, we observe the possible emergence of multiple epidemic waves. Notably, this occurs precisely in the cases where behavioural changes effectively prevent a large proportion of infections, namely for $q=0$ and $q=0.1$. On the other hand, for $q=0.2$, although the adoption of protective measures still contributes to containing the epidemic, a larger outbreak is observed compared to the cases $q=0$ and $q=0.1$, and no oscillatory behaviour or quasi-stationary regimes arise.

Finally, in \Cref{fig:protection_05} we compare the outputs of models \eqref{VBsystemAP} and \eqref{fastvarreduced} for $l=0.25$ and $l=1$ (hence with $\rho =2$ ad $\rho =5.12$, respectively) and $q=0.5$ (i.e.\ a 50\% probability of protection failure). In this case, we observe, in both the upper and lower rows, that the rates at which individuals change their behaviour in response to information do not significantly affect the overall dynamics. Nevertheless, both scenarios show that the model with protective behaviour leads to a smaller final size compared to the model without protection ($q=1$). In particular, the model with $l=1$ predicts a smaller final size than the model with $l=0.25$.

\section{Behaviour-induced dynamics in the low attack rate regime}\label{Sec:lowattackratio}

In this section, we investigate the dynamics of model \eqref{slowINFsystem} in a scenario where interventions and behavioural changes have successfully prevented a large proportion of infections before herd immunity is achieved and seasonality becomes dominant. In this setting, the susceptible populations are still sufficiently large, with $I_H\ll S_H$ and $I_M\ll S_M$, and only a small amount of information $J$ is left in the system (for instance, after the first epidemic peak has occurred)~\citep{zhang2023renewal}.
To model this \emph{low attack rate} scenario (see~\cite{dOnofrioLAR2021}), we assume that there exists $\bar t\ge 0$ such that
$I_H(t), I_M(t), J(t)\in \mathcal{O}(\varepsilon)$, for $0<\varepsilon\ll 1$ for all $t\ge \bar t$, while $S(t)\in\mathcal{O}(1)$, and the functions
$a,\ w$ satisfy $a(x)=\bar a(\varepsilon^{-1} x)$ and $w(x)=\bar w(\varepsilon^{-1} x)$. 
Note that the latter condition implies that information indicating that a very small fraction of the population has been infected is sufficient to induce a highly protective behaviour. 
We then introduce rescaled variables $\bar I_H = \varepsilon^{-1} I_H $ and so on.
With a slight abuse of notation, we retain the same symbols to avoid introducing additional variables or parameters (for instance, $\alpha$ will denote $\alpha/\varepsilon$). The resulting model reads
\begin{equation*}
\left\{\setlength\arraycolsep{0.1em}
\begin{array}{rl} 
	\dot{S}_H &= -\beta_{H\leftarrow M}\rho\varepsilon I_M h(p(J), q)S_H ,\\[3mm]
	\varepsilon\dot{I}_H &= \beta_{H\leftarrow M} \rho \varepsilon I_M h(p(J), q)S_H - \gamma  \varepsilon I_H  ,\\[3mm]
	\varepsilon\dot{I}_M &= \beta_{M\leftarrow H}(1- \varepsilon I_M) h(p(J), q) \varepsilon I_H  -\mu \varepsilon I_M,\\[3mm]
	\varepsilon\dot Z_1 &=k \varepsilon I_H -k \varepsilon Z_1, \\[3mm]
	\varepsilon\dot Z_i &= k \varepsilon Z_{i-1} - k \varepsilon Z_i, \qquad i=2,\dots, n-1,\\[3mm]
	\varepsilon \dot J &=k\varepsilon Z_{n-1} -k \varepsilon J,
\end{array} 
\right.
\end{equation*}
Then, by simplifying $\varepsilon$ in every equation except the first one and taking the limit $\varepsilon\to 0^+$, we obtain the corresponding layer system
\begin{equation}\label{outbreak}
\left\{\setlength\arraycolsep{0.1em}
\begin{array}{rl}
	\dot{I}_H &= \beta_{H\leftarrow M}\rho I_M h(p(J), q) S_H - \gamma  I_H  ,\\[3mm]
	\dot{I}_M &= \beta_{M\leftarrow H}h(p(J),q) I_H -\mu I_M,\\[3mm]
	\dot Z_1 &=k I_H -k Z_1 \\[3mm]
	\dot Z_i &= k Z_{i-1} - k  Z_i , \qquad i=2,\dots, n-1,\\[3mm]
	\dot J &=kZ_{n-1} -k  J,
\end{array} 
\right.
\end{equation}
with $\dot S_H=0$. Observe that the resulting systems resemble the model studied in~\cite{zhang2023renewal}, and subsequently extended in~\cite{ando2025}, for outbreaks of a directly transmitted infection under a low-attack-rate assumption.
In the following section, we will investigate the existence and stability of the equilibria of \eqref{outbreak} to gain insights into the transient dynamics of models \eqref{VBsystemAP} and \eqref{slowINFsystem}.

\subsection{Equilibria}
In this section, we investigate stationary solutions of system \eqref{outbreak}. Convergence to these equilibria corresponds to a rapid ``collapsing'' of the fast--slow system onto its critical manifold, in whose neighbourhood the slow dynamics play a central role. Characterising the existence of such equilibria and the associated convergence allows us to describe the transient behaviour of the system, even though the asymptotic dynamics inevitably leads to disease extinction, as shown in \Cref{extinction_1group}. 

We observe that the equilibria of system \eqref{outbreak} are all points $(I_H^*, I_M^*, J^*)\in \R^3_{\geq0}$ that satisfy the system of nonlinear equations  
\begin{equation}\label{equilibria}
\left\{\setlength\arraycolsep{0.1em}
\begin{array}{rl} 
	I_H^*&=\cfrac{\beta_{H\leftarrow M}}{\gamma}\rho S_H h(p(J^*), q)I_M^*,\\[3mm]
	I_M^* &= \cfrac{\beta_{M\leftarrow H}}{\mu} h(p(J^*),q) I_H^*,\\[5mm]
	Z_i^*&=J^* =I_H^*,\qquad i=1,\dots, n-1.
\end{array} 
\right.
\end{equation}
Substituting the second and third equations of \eqref{equilibria} into the one for $I_H^*$, we obtain the nonlinear equation
\begin{equation} \label{nonlineareq}
I_H^* \left(\hat R_e^2 h(p(I_H^*),q)^2 -1\right)=0,
\end{equation}
where 
\begin{equation}\label{Re}
\hat R_e\coloneqq  \sqrt{\cfrac{\beta_{H\leftarrow M}\beta_{M\leftarrow H}}{\gamma\mu}\rho S_H}
\end{equation}
is the \emph{Effective Reproduction Number} (ERN, i.e.\ the reproduction number when the population is not fully susceptible or partially immune~\citep{pellis2022}) in the absence of protective behaviour ($p=0$ or $q=1$) and non-competent hosts ($l=0$).
From \eqref{nonlineareq}, it follows that \eqref{outbreak} admits the DFE $I_H^*=I_M^*=J^*=0$. All other equilibria are determined by the solutions of the nonlinear equation
\begin{equation}\label{squared:nonlineq} 
\hat R_e^2h(p(I_H^*),q)^2=1\,.
\end{equation}
Note that, since $\hat R_e>0$ and $h(\cdot, \cdot)\ge 0$, the analysis reduces to studying the equation
\begin{equation}\label{condh} 
h(p(I_H^*),q)=\cfrac{1}{\hat R_e}\,.
\end{equation}
We have the following result:
\begin{theorem}\label{EEteo}
Let \Cref{ass:aw} hold and $q\in [0, 1)$. There exists a positive equilibrium of \eqref{outbreak} if and only if
\begin{equation}\label{condEE}
	1+l<\hat R_e<1+\frac{l}{q},
\end{equation}
and
\begin{equation}
	p(0) < p^* < \lim_{x \to +\infty}p(x),
	\label{cond2_EE}
\end{equation}
where 
\begin{equation}\label{eta}
	p^*\coloneqq \cfrac{\hat R_e-(1+l)
	}{\hat R_e-1}\ \cfrac{1}{1-q}\ .
\end{equation}
Furthermore, if such an equilibrium exists, then it is unique  and is given by
\begin{equation}\label{EE} 
	I_H^*= p^{-1}(p^*), \quad I_M^* = \cfrac{\beta_{M\leftarrow H}}{\mu}\ h(p(I_H^*),q) I_H^*,\quad J^* =I_H^*,
\end{equation}
recalling that $p(x)=a(x)/\left(a(x)+w(x)\right)$.
\end{theorem}
\begin{proof}
Note that $0<p^*<1$ by condition \eqref{condEE}. Furthermore, from \Cref{ass:aw} it follows that $p(\cdot)$ is strictly increasing, and hence \eqref{cond2_EE} implies the existence of a unique $I_H^*$ such that $p(I_H^*)=p^*$.
A straightforward computation, recalling the definition \eqref{def_H}, shows that $p^*$ satisfies $h(p^*,q)=1/\hat R_e$, so that \eqref{condh} holds. Uniqueness then follows from the fact that $p^*$ is the only solution of $h(p,q)=1/\hat R_e$, together with the injectivity of $p(\cdot)$.
\end{proof}
\begin{remark}
Since $f(p):=h(p,q)$ is a decreasing function, \eqref{cond2_EE} is equivalent to
\begin{equation}
	\lim_{x \to +\infty}h(p(x),q)< h(p^*,q)=\cfrac{1}{\hat R_e}< h(p(0),q) \iff  \hat R_e \lim_{x \to +\infty}h(p(x),q) < 1 < \hat R_e h(p(0),q).
	\label{cond_syntEE}
\end{equation}
Since 
$ q/(q+l)=h(1,q) \le \lim_{x \to +\infty}h(p(x),q) $ and $h(p(0),q) < h(0,q) = 1/(1+l)$, it follows that \eqref{cond_syntEE} implies \eqref{condEE}. Therefore, Theorem \ref{EEteo} can equivalently be restated by saying that an endemic equilibrium of \eqref{outbreak} exists if and only if \eqref{cond_syntEE} holds.

Note that, while $J(t)$ defined in \eqref{VBsystemAP} satisfies $0\le J(t) \le 1$, as proved in Proposition \ref{prop_ex_uniq2}, $J(t)$ solution of \eqref{outbreak} is a rescaling of the former, and we can only assume $J(t) \ge 0$; this explains why we consider $\lim_{x \to +\infty}p(x)$ instead of $p(1)$ in \eqref{cond2_EE} and \eqref{cond_syntEE}.
\end{remark}
In the following, in accordance with~\cite{ando2025}, we refer to the non-trivial equilibrium in \eqref{EE} as the ``\emph{Established Equilibrium}'' (EE) rather than the ``\emph{Endemic Equilibrium}'', as it represents a stationary situation reached by the system only due to behavioural changes during an outbreak, rather than in a truly endemic setting. 

The two conditions of Theorem \ref{EEteo} show that the existence of the EE depends, on the one hand, on the relation between $\hat R_e$, $l$, and $q$, on the other hand, on the function $p(\cdot)$, which determines how humans react to information on prevalence. In particular, the more effective the protection, the less restrictive the upper bound in \eqref{condEE} for the existence of an EE. On the contrary, the larger $l$, the more restrictive the lower bound in \eqref{condEE}.
From this observation, for small $l\ne0$ and large $l/q$, we can expect the outbreak to persist for a longer time.
Note that this situation actually corresponds to a mosquito-borne epidemic in a region where high human density and low availability of alternative blood sources do not limit transmission. In particular, as some individuals adopt highly effective protective measures, mosquitoes concentrate their bites on unprotected individuals, thereby increasing the likelihood of a full transmission cycle(vector$\rightarrow$human$\rightarrow$vector).

Finally, we observe that, for $l=0$, one has $h\equiv 1$. Hence, we have the following:
\begin{corollary}
If $l=0$, then system \eqref{outbreak} admits the DFE only.   
\end{corollary}

\subsection{Stability of equilibria}
This section aims to investigate the local stability of the equilibria of \eqref{outbreak}. 
We do this by linearising \eqref{outbreak} around a given equilibrium $(\bar I_H,\ \bar I_M,\ \bar J)$ and analysing the sign of the real parts of the eigenvalues of the corresponding Jacobian. 
First, we show that the DFE is Locally Asymptotically Stable (LAS) when $\hat R_e h(p(0), q)<1$, while it is unstable when $\hat R_e h(p(0), q)<1$, independently of the choice of $K$ in \eqref{infindex} and \eqref{erlang}. Then, we consider the stability of the EE, showing that it depends on the specific form of the memory kernel. 
In particular, for selected choices of $K$, we show that the EE may either be LAS or lose its stability via Hopf, hence ensuring the possible emergence of self-sustained oscillations even in an outbreak scenario.

Let 
\begin{equation}\label{f=h}
f(x)\coloneqq h(p(x), q),\qquad  x\in [0,+\infty),
\end{equation}
so that $f'(x)= \partial_p h(p(x), q)\cdot p'(x)$. 
The linearisation of \eqref{outbreak} around an equilibrium $(\bar I_H,\ \bar I_M,\ \bar J)$ reads
\begin{equation}\label{lin:outbreak}
\left\{\setlength\arraycolsep{0.1em}
\begin{array}{rl} 
\dot{I}_H &= \beta_{H\leftarrow M}\rho S_Hf(\bar J)I_M+\beta_{H\leftarrow M}\rho\bar I_M S_Hf'(\bar J)J - \gamma  I_H  ,\\[3mm]
\dot{I}_M &= \beta_{M\leftarrow H} f(\bar J)I_H+\beta_{M\leftarrow H} f'(\bar J) \bar I_H J-\mu I_M,\\[3mm]
\dot Z_1 &=k I_H -k Z_1 \\[3mm]
\dot Z_i &= k Z_{i-1} - k  Z_i ,\qquad \qquad \qquad i=2,\dots, n-1,\\[3mm]
\dot J &=kZ_{n-1} -k  J.
\end{array} 
\right.
\end{equation}
Then, for $K$ as in \eqref{erlang}, by proceeding as in \Cref{app:char}, one obtains the following characteristic equation for the eigenvalues of the Jacobian relevant to \eqref{lin:outbreak}:
\begin{align}\label{gen_char}
\lambda^2+\lambda(\gamma+\mu)+\gamma\mu\left[1 -\hat R_e^2f^2(\bar J)\right]- \hat K(\lambda)f(\bar J)f'(\bar J)\gamma \hat R_e^2(\lambda + 2\mu) \bar I_H=0,\qquad \Re(\lambda)>-k,
\end{align}
where $\hat K$ denotes the Laplace transform of $K$, i.e.
\begin{equation}\label{laplace_erl}
\hat K(\lambda)\coloneqq\displaystyle\int_{0}^{+\infty}{\e}^{-\lambda\theta} K(\theta)\dd \theta=\left(\frac{k}{\lambda+k}\right)^n,\qquad \Re(\lambda)>-k.
\end{equation}
In the following sections, we specialise \eqref{gen_char} for the case of the DFE and the EE.

\begin{remark}
The derivation of \eqref{gen_char} is provided in \Cref{app:char}, where we apply the relevant theory for delay equations~\citep{diekmann2012delay} with, possibly, infinite delay~\citep{DiekmannGyllenberg2012Blending}. As a result, equation \eqref{gen_char} holds for any memory kernel $K$, not necessarily restricted to Erlang-distributed.
\end{remark}

\subsubsection{Stability of the DFE}
Let us consider the DFE $(\bar I_H, \bar I_M, \bar J)=(0, 0, 0)$. Then, the characteristic equation in \eqref{gen_char} becomes
\begin{equation}\label{chareq_DFE}
\lambda^2+\lambda(\gamma+\mu)+\gamma\mu\left[1 -\hat R_e^2f^2(0)\right]=0,\qquad \Re(\lambda)>-k.    
\end{equation}
Recalling equation~\eqref{f=h}, we have the following result.
\begin{theorem}
The DFE is LAS when $\hat R_eh(p(0), q)<1$ and unstable when $\hat R_eh(p(0), q)>1$, independently of $K$ (i.e. not restricted to the one in \eqref{erlang}). In particular, in the latter case, the characteristic equation \eqref{chareq_DFE} has exactly two real roots, one negative and one positive.
\end{theorem}

\begin{proof}
Being $\gamma,\ \mu>0$, Descartes' rule of signs 
ensures that the DFE is stable when $1-\hat R_eh(p(0), q)>0$, while it is unstable when $1-\hat R_eh(p(0), q)<0$. \end{proof}
Comparing this condition with \eqref{cond_syntEE}, we see that, when the DFE is stable, no EE exists. Conversely, when the DFE is unstable, one additional condition is still needed to guarantee the existence of the EE.

\subsubsection{Stability of the EE}\label{sec:stab_EE}
Let us consider the EE defined as in \eqref{EE}.
Using \eqref{squared:nonlineq} and \eqref{laplace_erl}, the characteristic equation \eqref{gen_char} reads
\begin{equation}\label{char_gamma_div}
\lambda^2+\lambda(\gamma+\mu)+\left(\frac{k}{\lambda+k}\right)^n\gamma\delta\left(\lambda   + 2\mu \right)=0, \qquad \Re(\lambda)>-k,
\end{equation}
for $\delta\coloneqq -f'(I_H^*)\hat R_e I_H^*>0$. 
Note that, in this case, the stability of the equilibrium depends on the particular choice of $n$ and $k$. 
In particular, from \eqref{char_gamma_div}  we are led to study the roots of 
\begin{align}\label{CharEqGenErlang}
\lambda^2(\lambda+k)^n+\lambda(\gamma+\mu)(\lambda+k)^n+ \lambda (k^n\gamma \delta)+2 k^n\gamma \delta\mu=0.    
\end{align}
It is worth observing that, if $k\to +\infty$ and $n$ is fixed, one has $\varphi\coloneqq n/k\to0^+$ (that is, the information is instantaneous) and the characteristic equation reduces to 
\begin{equation}\label{char_eq_inst}
\lambda^2+\lambda\left(\gamma+\mu+ \gamma \delta\right)+ 2\mu\gamma \delta=0.
\end{equation}
This case corresponds to the case in which $K$ is a Dirac delta concentrated at $0$, so that $\hat K\equiv 1$. 
On the other hand, if $n,\ k\to +\infty$ in such a way that $\varphi\coloneqq n/k$ is constant, then the kernel concentrates at $t-\varphi$, i.e.\ the memory becomes discrete with delay $\varphi$. In this case, the characteristic equation takes the form
\begin{align*}
\lambda^2+\lambda(\gamma+\mu)+ \gamma\delta\lambda{\e}^{-\lambda\varphi}+\gamma\mu\delta{\e}^{-\lambda\varphi}=0.
\end{align*}
We consider the following cases: instantaneous information, exponentially fading memory, and Erlang-2 distributed memory.\\

\emph{\textbf{Instantaneous information}}: the characteristic equations read as in equation \eqref{char_eq_inst}.
Then, from \eqref{der:h}, the Routh-Hurwitz criterion implies that the EE is LAS whenever it exists.\\

\emph{\textbf{Exponentially fading memory}}: we assume that $K(\theta)=k{\e}^{-k\theta}$, corresponding to $n=1$ in \eqref{erlang}. Then, equation \eqref{CharEqGenErlang} reduces to 
\begin{equation}\label{char_exp}
\lambda^2(\lambda+k)+\lambda(\gamma+\mu)(\lambda+k)+ k\gamma\delta \left(\lambda + 2\mu \right)=0.
\end{equation}
Applying the Routh--Hurwitz criterion for third-order polynomials (see \Cref{comp_EE_exp} for details), we obtain the following result.
\begin{proposition}\label{propexp}
Let the assumptions of \Cref{EEteo} hold, and let $n=1$ in \eqref{erlang}. Then the EE is LAS if and only if
\begin{equation}\label{exp:condLAS}
k>-\frac{\gamma\delta(\gamma-\mu)+(\gamma+\mu)^2}{\gamma+\mu+\gamma\delta}.
\end{equation} 
If equality holds, the EE undergoes a Hopf bifurcation, and the characteristic equation \eqref{char_exp} has one negative real root and a pair of purely imaginary conjugate roots. If the opposite inequality holds, then the EE is unstable, and the characteristic equation \eqref{char_exp} has one negative real root and a pair of complex conjugate roots with positive real parts.
\end{proposition}
Since \eqref{der:h} holds and $k>0$, condition \eqref{exp:condLAS} is automatically satisfied whenever $\gamma+k\ge\mu$. In particular, this is always true if $\gamma\ge\mu$. This assumption is biologically reasonable, as the average human infectious period $1/\gamma$ is typically shorter than the average mosquito lifespan $1/\mu$. 
On the other hand, if $\gamma<\mu$, then condition \eqref{exp:condLAS} provides a (possibly positive) lower bound for $k$. This yields the following corollary.
\begin{corollary}
Let the assumptions of \Cref{EEteo} hold, and let $n=1$  in \eqref{erlang}.  If $\gamma\ge \mu$, then the EE is LAS for all $k>0$.
\end{corollary}

\emph{\textbf{Erlang-2 distributed memory}}: we assume that $K$ follows an Erlang distribution of order 2, i.e. we take $n=2$ in~\eqref{erlang}. Thus, 
the equation \eqref{CharEqGenErlang} becomes 
\begin{align}\label{char:erl}
\lambda^2(\lambda+k)^2+\lambda(\gamma+\mu)(\lambda+k)^2+ k^2\gamma \delta\left(\lambda   + 2\mu \right) =0. \end{align}
Then, applying the Routh-Hurwitz criterion to a fourth-order polynomial (see \Cref{comp_EE_Erlang2} for the details), we obtain the following result.
\begin{proposition}\label{properlang}
Let the assumptions of \Cref{EEteo} hold, and let $n=2$ in \eqref{erlang}. Then, there exist $k_+, \,k_-\ge0$ with $k_+>k_-$ such that the EE is LAS for $k>k_+$, undergoes a Hopf bifurcation at $k=k_+$, and is unstable for $k\in (k_-, k_+)$. In particular, $k_+$ is the largest positive real root of
$$p(k)=\tilde A k^3+\tilde Bk^2+\tilde C k+\tilde D,$$
where
\begin{align*}
\tilde A&\coloneqq 2(\gamma+\mu+\gamma \delta)>0,\\
\tilde B&\coloneqq 4(\gamma+\mu)^2+\gamma\delta(3\gamma-5\mu-\gamma\delta),\\
\tilde C&\coloneqq 2(\gamma+\mu)[(\gamma+\mu)^2+\gamma\delta(\gamma-3\mu)],\\
\tilde D&\coloneqq 2\gamma\mu\delta(\gamma+\mu)^2<0,
\end{align*}
while $k_-$ is either $0$ or a positive root of $p(k)=0$ such that $p'(k_-)<0$.
\end{proposition}

\begin{remark}
Note that, once the epidemiological parameters $\beta_{H\leftarrow M},\, \beta_{M\leftarrow H},\, \rho,\, \gamma,$ and $\mu$ are fixed, the thresholds $k_-$ and $k_+$ are fully determined by $\delta$, which itself depends on $l$, $q$, $a$, and $w$.
\end{remark}

In this section, we showed that a sufficiently large delay $\varphi$ can destabilise the EE via a Hopf bifurcation. This suggests that self-sustained epidemic oscillations may arise even in the absence of additional mechanisms such as demography or waning immunity. 
Here, we restricted the analysis to Erlang-distributed memory kernels with shape parameter $n=1$ or $n=2$; see \eqref{erlang}.
As noted above, the analysis may be extended to more general kernels, although this would likely result in considerably more involved computations~\citep{ando2025, zhang2023renewal}.

\subsection{Numerical results}

\begin{figure}[hbt!]
\begin{center}
\includegraphics[width=1.\linewidth]{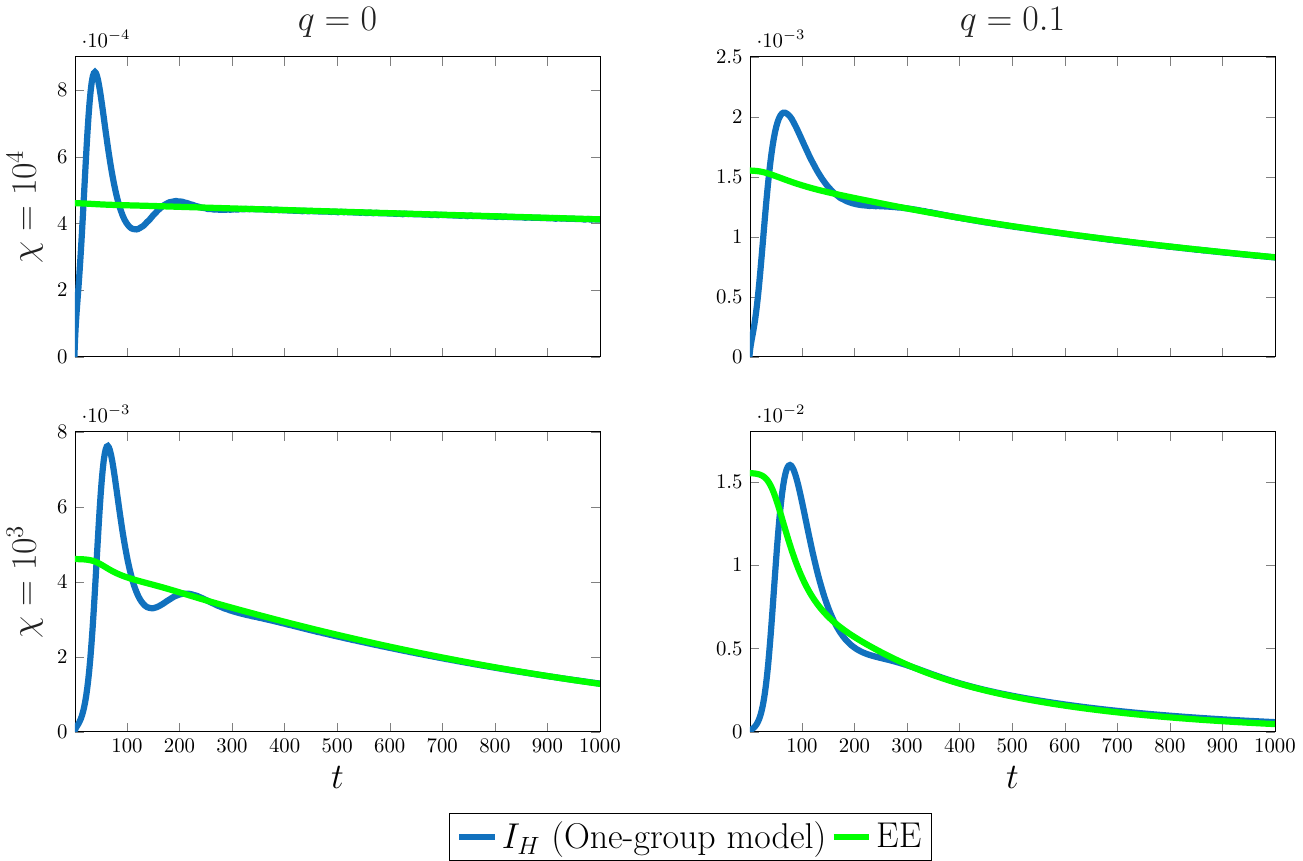}
\caption{Time evolution of $I_H$ (blue) for model \eqref{slowINFsystem} (one-group model) and EE (green) as a function of $S_H$. Model parameters and initial conditions are as in \Cref{sec:numerics}, with $l=0.25$ and $\rho =2$. The upper row corresponds to $\chi =10^4$, the lower row to $\chi =10^3$; the left panels correspond to $q=0$ and the right panels to $q=0.1$. We take $n=1$ and $\varphi=30$ days in \eqref{erlang}.
	\label{fig:one_group_EE_exp_phi_30}}
	\end{center}
\end{figure}

\begin{figure}[hbt!]
\begin{center}
\includegraphics[width=1.\linewidth]{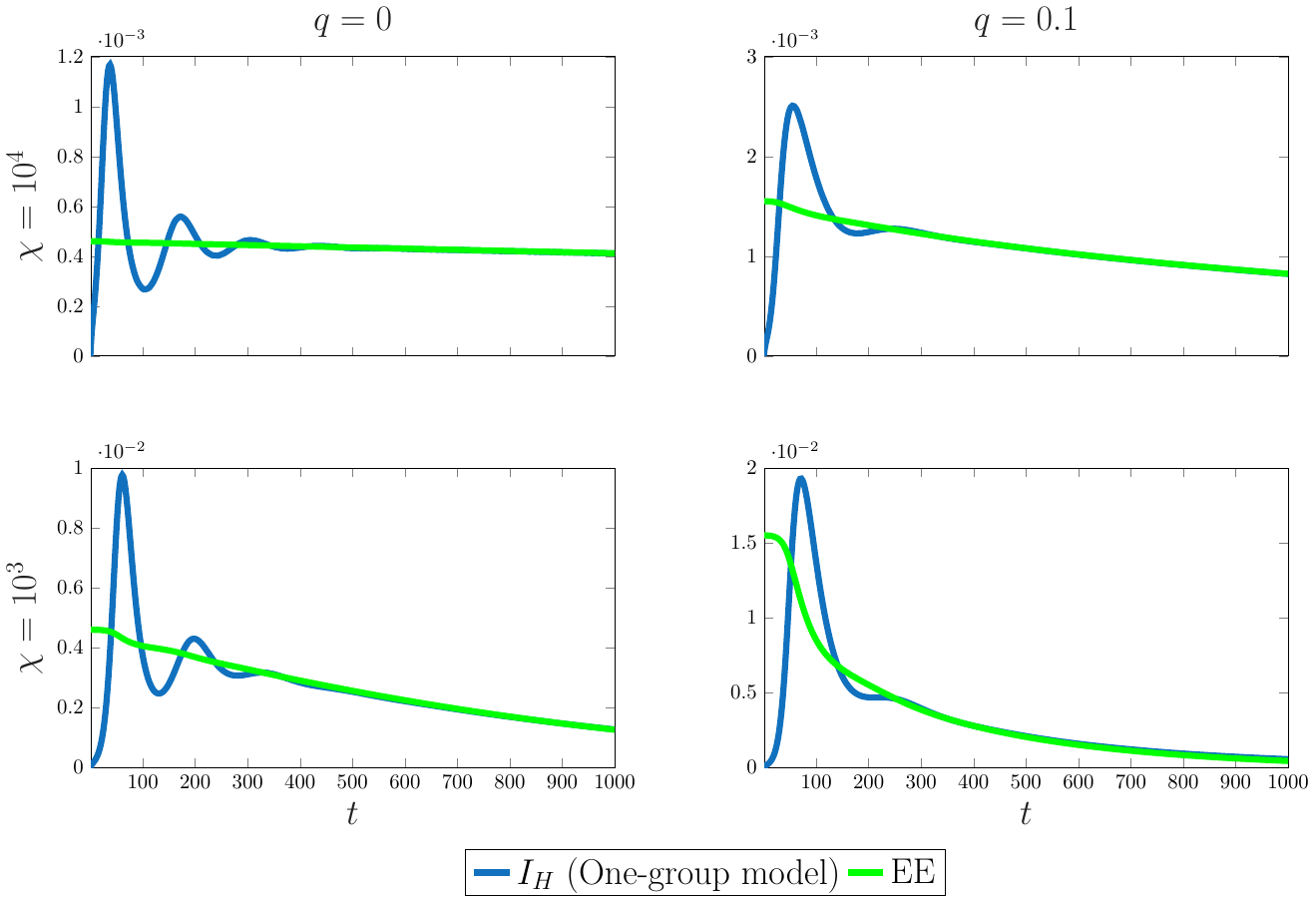}
\caption{Time evolution of $I_H$ (blue) for model \eqref{slowINFsystem} (one-group model) and EE (green) as a function of $S_H$. Model parameters and initial conditions are as in \Cref{sec:numerics}, with $l=0.25$ and $\rho =2$. The upper row corresponds to $\chi =10^4$,  the lower row to  $\chi =10^3$; the left panels correspond to $q=0$ and the right panels to $q=0.1$. We take $n=2$ and $\varphi=30$ days in \eqref{erlang}. 
	\label{fig:one_group_EE_erl2_phi_30}}
	\end{center}
\end{figure}

\begin{figure}[hbt!]
\begin{center}
\includegraphics[width=1.\linewidth]{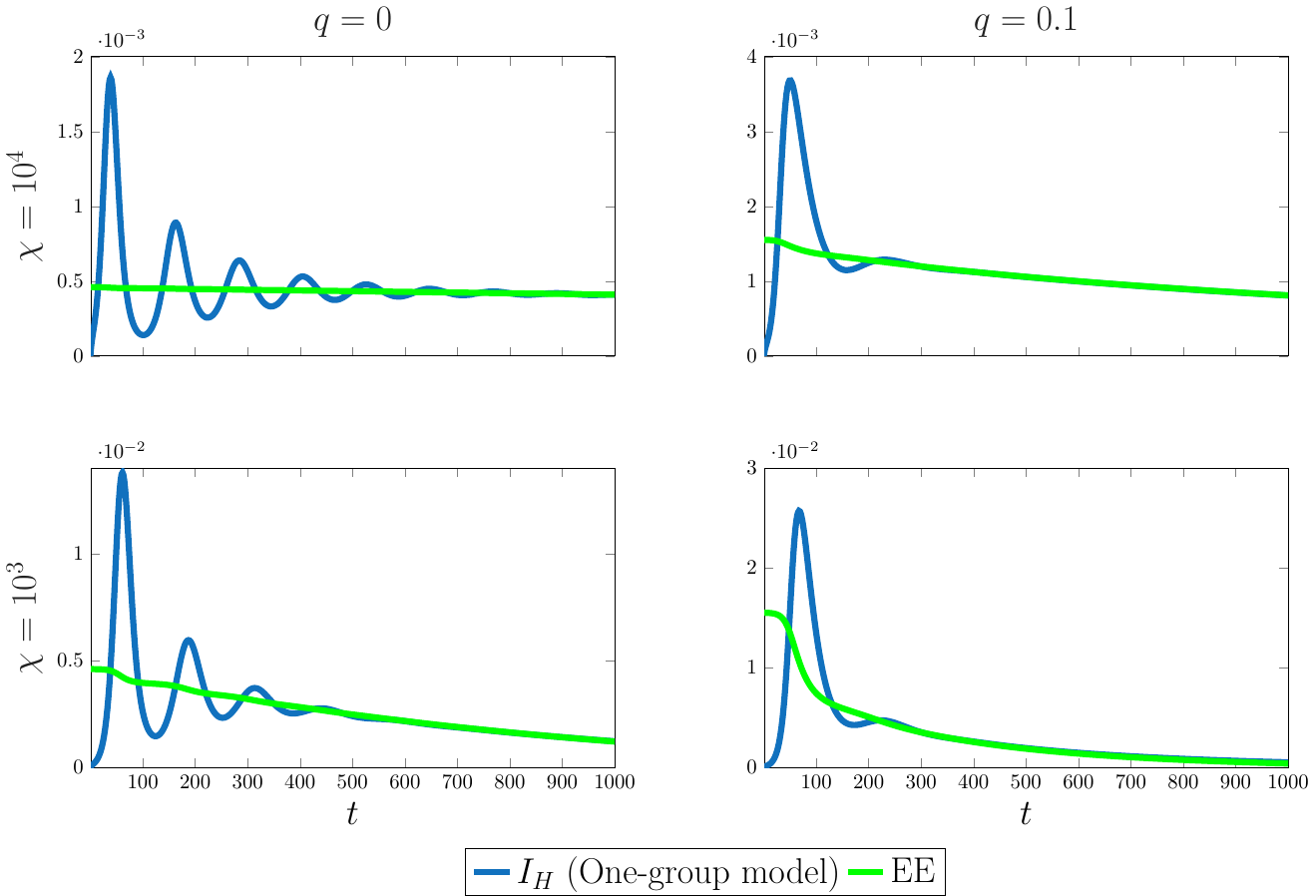}
\caption{Time evolution of $I_H$ (blue) for model \eqref{slowINFsystem} (one-group model) and EE (green) as a function of $S_H$. Model parameters and initial conditions are as in \Cref{sec:numerics}, with $l=0.25$ and $\rho =2$. The upper row corresponds to $\chi =10^4$,  the lower row to  $\chi =10^3$; the left panels correspond to $q=0$ and the right panels to $q=0.1$. We take $n=5$ and $\varphi=30$ days in \eqref{erlang}. 
	\label{fig:one_group_EE_erl5_phi_30}}
	\end{center}
\end{figure}

\begin{figure}[t]
\begin{center}
\includegraphics[width=.9\linewidth]{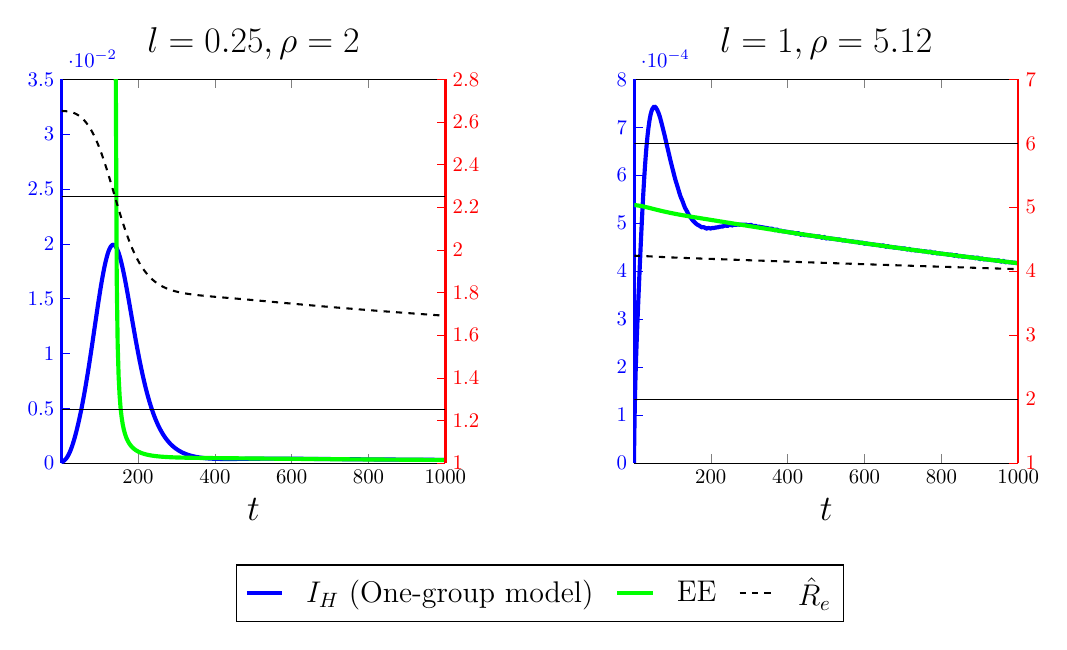}
\caption{Time evolution of $I_H$ (solid blue) for model \eqref{slowINFsystem} (one-group model), EE (solid green) and $\hat R_e$ (dashed black) as functions of $S_H$. Model parameters and initial conditions are as in \Cref{sec:numerics}, with $l=0.25$ and $\rho =2$ (left), $l=1$ and $\rho =5.12$ (right). We take $\chi =10^4$,
	$q=0.2$, with $n=1$ and $\varphi=30$ days in \eqref{erlang}. For each panel, the left vertical axis (blue) corresponds to $I_H$ and EE, while the right vertical axis (red) corresponds to $\hat R_e$. The upper and lower solid black horizontal lines represent the thresholds $1+l$ and $1+l/q$ in \eqref{condEE}, respectively. Note that, in the left panel, the EE only exists once $\hat R_e$ crosses the threshold value identified in Theorem \ref{EEteo}.
	\label{fig:one_group_q=0.2}}
	\end{center}
\end{figure}

In this section, we validate the results of \Cref{Sec:lowattackratio} by comparing the outputs of model \eqref{fastvarreduced} with those of \eqref{outbreak}. We use the same parameter values as in \Cref{sec:numerics}, with  $\varphi =30$ days. In Figures \ref{fig:one_group_EE_exp_phi_30}--
\ref{fig:one_group_EE_erl5_phi_30}, we show the time evolution of $I_H$  (blue) and the EE as a function of $S_H$ (green, with $S_H$ varying in time, see \eqref{EE}) for $n=1,\,2,\,5$, respectively. We fix $l=0.25$ and $\rho =2$, and consider $q=0$ (left panels) and $q=0.1$ (right panels). For each case, we compare the model outputs for $\chi =10^4$ (upper row) and $\chi =10^3$ (lower row).  

In all cases, the solution $I_H$ of model \eqref{fastvarreduced} approaches the EE predicted by model \eqref{outbreak}, which can be shown to be LAS for these parameter values whenever it exists. 
In particular, for $q=0$, increasing $n$ leads to more pronounced oscillations around the EE before the solution ultimately collapses onto this manifold, for both $\chi =10^4$ and $\chi =10^3$. These oscillatory behaviours are instead much less pronounced (and nearly absent) for the case $q=0.1$.

We now recall that, in \Cref{sec:numerics}, we observed that in all experiments in Figures \ref{fig:one_group_EE_exp_phi_30}-- 
\ref{fig:one_group_EE_erl5_phi_30}, the choice $q=0.2$ with $l=0.25$ and $\rho =2$ led to larger outbreaks compared to the other cases (except for $q=0.5$), and no ``quasi steady state'' was observed. To explain this, in \Cref{fig:one_group_q=0.2} we plot  $I_H$ (blue) and the EE (green) as functions of time (with EE depending on $S_H$), together with $\hat R_e$ as a function of $S_H$ (dashed line; see \eqref{Re}), for $\chi =10^4$. We consider both $l=0.25$, $\rho=2$ (left panel) and  $l=1$, $\rho =5.12$ (right panel). 
In the former case, at the beginning of the outbreak, $\hat R_e$ exceeds the upper threshold $1+l/q$ in \eqref{condEE}, which is a necessary condition for the existence of the EE (in the figure, the solid black horizontal lines represent the thresholds $1+l$ and $1+l/q$). As a consequence, a larger outbreak occurs, reaching its peak when $\hat R_e=1+l/q$, after which the epidemic declines and eventually dies out.
In contrast, for $l=1$ and $\rho=5.12$, the condition $\hat R_e<1+l/q$ holds, ensuring the existence of the EE for model \eqref{outbreak}. In this case, the equilibrium is LAS, and consequently, the solution $I_H$ converges to it.

\section{Conclusions and outlook}\label{sec:conclusions}

In this paper, we investigated the impact of protective behaviour on an SIR-SI host--vector compartmental model. We focused on single-outbreak scenarios (thus neglecting host demography and waning immunity) so that the systems under study admit only the DFE as a non-trivial equilibrium. Nevertheless, we were able to characterise the transient dynamics under various modelling assumptions.

In Section \ref{sec:model}, we first analysed a model in which individuals adopt a fixed behaviour (either protecting or not protecting themselves from mosquito bites). Such models have been extensively studied in the literature~\citep{dye1986,dye1988,miller2013effects}, where it has been shown that imperfect or partial protection strategies may increase the basic reproduction number $R_0$, thereby enhancing the risk of an outbreak. In contrast, our analysis (Proposition \ref{propineqRatio}) shows that imperfect protection in a fraction of the population may decrease $R_0$ whenever a function $F(p,q)$, depending on the fraction $p$ of protected individuals and on the protection leakage $q$, exceeds the parameter $l$, representing the fraction of mosquito bites on non-human hosts. To the best of our knowledge, this result is new, although consistent with previous findings. For instance, in~\cite{miller2016risk}, threshold conditions on $p$ were derived to determine whether protective behaviour increases or decreases $R_0$. In particular, $R_0$ increases when protection diverts mosquito bites towards non-protected individuals~\citep{killeen2007exploring, moore2007mosquitoes}, whereas it may decrease if protective measures force mosquitoes to spend additional time attempting to bite protected hosts, indirectly benefiting non-protected individuals. This interpretation is consistent with viewing $l$ as a measure of the spatial separation between hosts (see \Cref{remark_no_handling}).

We then considered a model in which behavioural changes depend on information about the epidemic. In this setting, we proved (\Cref{demers}) that the reproduction number is smaller when individuals switch between protected and unprotected behaviour than in the case of fixed behaviour, for the same values of $p$ and $q$. While this effect had been observed numerically in~\cite{demers2018dynamic}, a rigorous proof of \eqref{RcminR0} was, to our knowledge, not previously available. Biologically, this result shows that the concentration of mosquito bites on non-protected individuals (see~\cite{dye1986, miller2016risk, miller2013effects}) is mitigated when individuals switch between protected and unprotected behaviour, even when such changes occur at low rates~\citep{demers2018dynamic}.

When it is assumed that the rate of behavioural changes is much faster than epidemic dynamics (a common assumption in this class of models; see~\cite{bulai2024geometric, della2024geometric, poletti2009spontaneous}), it is possible to separate the time scales of the two processes by applying methods from geometric singular perturbation theory. This leads to the epidemic model \eqref{slowINFsystem}, which includes the information index while treating the host population as homogeneous.
Note that \cite{buonomo2014modeling} proposed an SIR--SI model similar to \eqref{slowINFsystem} in an endemic setting for a malaria-like infection, accounting for human behavioural responses to public health campaigns promoting the use of bed nets. In that work, the mosquito biting rate (interpreted as the human--mosquito contact rate) is assumed to depend on past prevalence through an information index, as well as on an \emph{effort} function describing the actions of the PHS in the health-promotion campaign. However, the analysis in~\cite{buonomo2014modeling} mainly focuses on numerical results and optimal control aspects for specific choices of the information kernel, with less emphasis on a detailed analytical study of the model dynamics.

Our analysis of model \eqref{slowINFsystem} (in Section \ref{Sec:lowattackratio}) shows that the inclusion of information-induced protective behaviour may lead to prolonged epidemic waves. After the initial peak, solutions rapidly approach a slow manifold along which the infected population decays very slowly. Numerical simulations confirm this behaviour, showing that an outbreak may last for a very long time, even in the absence of host demography and recruitment of new susceptibles. Moreover, numerical simulations indicate that convergence to this slow manifold persists even when the assumption of a clear separation of time scales is relaxed. This suggests that prolonged outbreaks are an intrinsic feature of information-dependent behavioural models. We stress, however, that the model assumes constant parameters; in realistic settings, seasonal effects (e.g., temperature-driven mosquito dynamics) may significantly shorten the duration of an outbreak.

Based on our analytical and numerical results, several extensions of the model can be considered. A natural first step would be to include demography and/or loss of immunity in the human population, which would allow for the existence of an Endemic Equilibrium under standard assumptions on the parameters. Another possible extension is to incorporate seasonality, for instance by accounting for seasonal fluctuations in the mosquito population (e.g., due to temperature and humidity); see, for instance, \cite{rocha2016}.
However, in this setting, even the computation of the BRN and CRN becomes more involved, as they are defined as the spectral radius of an infinite-dimensional next-generation operator rather than of a next-generation matrix (NGM)~\citep{Bacaer2006, inaba2019}, and typically require numerical approximations; see, e.g., \cite{Bacaer2007, bredadereggiripollR0}.
Additionally, it would be interesting to investigate models evolving on three time scales: a fast one associated with the spread of information, an intermediate one corresponding to disease dynamics, and a slow one related to human demography and/or loss of immunity.

Finally, we note that individual behaviour depends on personal opinions about the infection, which may in turn be influenced by information on the epidemic as well as by social interactions. Several authors have recently investigated the interplay between opinion dynamics and epidemic spread through mathematical models~\citep{Albi2025,Chang2025,Tyson2020}.
Without explicitly introducing a variable describing opinion, an extension of the behavioural model \eqref{modelbehH} might be obtained by including imitation-dynamics terms of the form $\theta_P(J)H_{NP}p$ and $\theta_{NP}(J)H_{P}(1-p)$, which describe a contagion of ideas among host individuals~\citep{bauch2005, manfredi2013book, WangStatistical}. Such contagion may be driven by individuals’ knowledge of past epidemics~\citep{della2024geometric}, or by payoff considerations~\citep{poletti2009spontaneous}. Accordingly, the equation for $H_P$ (and correspondingly for $H_{NP}$) takes the form
\begin{equation*}
H_P'=\left[a(J) + \theta_P(J)p\right]H_{NP} - \left[w(J)+\theta_{NP}(J)(1-p)\right]H_{P}.
\end{equation*}
In the context of vector-borne epidemics, similar approaches have been adopted in~\cite{asfaw2018impact} and, more recently, in~\cite{laxmi2022} to model the effect of bed-net usage in malaria-like diseases. The latter, in particular, focuses on both the use and misuse of Insecticide-Treated Nets (ITNs).
Both works deal with endemic scenarios: the former is based on imitation game dynamics with information-free behavioural changes, while the latter derives a related model in which behavioural changes are driven by payoff considerations linked to infection risk, which may depend on information about current prevalence, mosquito density, and seasonal effects (e.g., periodic replacement of bed nets). 
While~\cite{asfaw2018impact} provide analytical results on the model dynamics,~\cite{laxmi2022} mainly focus on reproduction numbers (and related considerations on optimal ITN usage) and numerical simulations, while also showing that information-induced imitation dynamics may generate recurrent epidemic waves. Interestingly,~\cite{laxmi2022} also discuss applications to malaria control in several African countries.
In the present work, we did not consider such mechanisms. Nevertheless, the approach developed here could be extended to models incorporating evolutionary game-like dynamics, although this would inevitably lead to more involved computations (see, e.g.,~\cite{della2024geometric}). This will be the subject of future work, together with the inclusion of additional factors relevant to the dynamics and control of mosquito-borne epidemics, and applications of the proposed behavioural models to data-informed real-world scenarios.

\section*{Acknowledgements}
This work was supported by the project ``One Health Basic and Translational Actions Addressing
Unmet Needs on Emerging Infectious Diseases'' (INF-ACT), BaC ``Behaviour and sentiment monitoring and modelling for outbreak control/BEHAVE-MOD'' (No. PE00000007, CUP I83C22001810007)
funded by the NextGenerationEU.
The authors are members of the \emph{Unione Matematica Italiana} (UMI) group ``\emph{Modellistica Socio-Epidemiologica}'' (UMI-MSE) and of the following groups of the \emph{Istituto Nazionale di Alta Matematica} (INdAM):  ``GNCS -- \emph{Gruppo Nazionale per il Calcolo Scientifico}'' (SDR), ``GNAMPA -- \emph{Gruppo Nazionale per l’Analisi Matematica e le sue Applicazioni}” (AP), and ``GNFM -- \emph{Gruppo Nazionale per la Fisica Matematica}'' (MS, CS).

\section*{Declarations}

\noindent \textbf{Data Availability} The paper does not analyse any data. The simulations are obtained using MATLAB 2025b. Relevant programs can be requested from the
authors.

\begin{appendices}

\section{Computation of \texorpdfstring{$\hat R_c$}{Rc}}
In this section, we provide additional details on the computation of $\hat R_c$ for model~\eqref{VBsystemAP} in \Cref{R0DynamicChanges}.
Observe that the inverse of $\hat \Sigma$ in \eqref{tranMat} is given explicitly by
\begin{equation*}
\hat \Sigma^{-1}\coloneqq 
\begin{pmatrix} 
	\dfrac{\gamma+a_0}{\gamma(a_0+w_0+\gamma)}  & a_0 & 0\\
	w_0 & \dfrac{\gamma + w_0}{\gamma(a_0+w_0+\gamma)}  & 0\\
	0 & 0 & \dfrac{1}{\mu} 
\end{pmatrix}.
\end{equation*}
Hence, the NGM $\hat K\coloneqq B(p,q)\hat \Sigma$, for $B$ defined as in \eqref{infection:matrix}, reads
\begin{equation*}
\hat K\coloneqq 
\begin{pmatrix} 
	0 & 0 & \dfrac{\rho\beta_{H\leftarrow M}}{\mu}\dfrac{qp_0}{c(p_0, q)+l}\\
	0 & 0 & \dfrac{\rho\beta_{H\leftarrow M}}
	{\mu}\dfrac{1-p_0}{c(p_0, q)+l}\\
	\dfrac{\beta_{M\leftarrow H}[q(a_0+\gamma)+w_0]}{\gamma(a_0+w_0+\gamma)[c(p_0, q)+l]} & \dfrac{\beta_{M\leftarrow H}[qa_0+w_0+\gamma]}{\gamma(a_0+w_0+\gamma)[c(p_0, q)+l]} & 0
\end{pmatrix}.
\end{equation*}
The characteristic polynomial of $\hat K$ reads
$$p(\lambda)=\lambda^3-\lambda\left[\frac{qp_0[q(a_0+\gamma)+w_0]+(1-p_0)[qa_0+w_0+\gamma]}{a_0+w_0+\gamma}\right]\frac{\beta_{M\leftarrow H}\beta_{H\leftarrow M}}{\gamma\mu}\frac{\rho }{c(p_0, q)+l},$$ for $\lambda \in \C$. From this expression, the formula in \eqref{hatR0beh} can be readily recovered.

\section{Derivation of the characteristic equation }\label{app:char}
In this section, for the sake of generality, we rewrite \eqref{outbreak} using the integral formulation of $J$ in \eqref{infindex}. This allows us to derive the characteristic equations needed to analyse the stability of the equilibria of \eqref{outbreak} by means of Laplace transform techniques, as in~\cite{ando2025}.
Consider
\begin{equation}\label{pre_linearisation}
\left\{\setlength\arraycolsep{0.1em}
\begin{array}{rl} 
	\dot{I}_H &= \beta_{H\leftarrow M} \rho I_M S_Hh(p(J), q) - \gamma  I_H  ,\\[3mm]
	\dot{I}_M &= \beta_{M\leftarrow H} h(p(J),q) I_H -\mu I_M,\\[1mm]
	J(t) &=\displaystyle\int_0^{+\infty} I_H(t-\theta) K(\theta)\dd \theta,
\end{array} 
\right.
\end{equation}
with $K$ as in \eqref{erlang}.
Note that the above system couples two ODEs for $I_H $ and $I_M$ with a delay equation with infinite delay for $J(t)$, since it depends on the history $(I_H)_t(\theta):=I_H(t+\theta)$, for $\theta \in \mathbb{R}_{\leq 0}$.
For this model, the natural choice of the history space is $C_\omega(\R_{\le 0}, \R)\times \R\times \R$, where $C_\omega(\R_{\le 0}, \R)$ denotes the space of continuous functions $\psi\colon \R_{\le0}\to \R$ such that $\lim_{\theta\to-\infty} \omega(\theta)\psi(\theta)\to 0$ with $\omega(\theta):= {\e}^{\nu \theta}$ for some $\nu>0$ chosen a priori~\citep{DiekmannGyllenberg2012Blending}. In this setting, the existence and uniqueness of solutions of \eqref{pre_linearisation}  follow from standard results; see~\cite{DiekmannGyllenberg2012Blending}.
Observe that \eqref{outbreak} and \eqref{pre_linearisation} share the same equilibria and are equivalent from the point of view of the stability analysis, although they differ in the underlying choice of the state space.

Now, we investigate the stability of equilibria of \eqref{pre_linearisation} by applying the principle of linearised stability for equations with infinite delay~\citep{DiekmannGyllenberg2012Blending}. For $f$ defined as in \eqref{f=h}, the linearisation of \eqref{pre_linearisation}  reads
\begin{equation}\label{lin:outbreak2}
\left\{\setlength\arraycolsep{0.1em}
\begin{array}{rl} 
	\dot{I}_H &= \beta_{H\leftarrow M} \rho S_Hf(\bar J)I_M+\beta_{H\leftarrow M}\rho \bar I_M S_Hf'(\bar J)J - \gamma  I_H  ,\\[5mm]
	\dot{I}_M &= \beta_{M\leftarrow H} f(\bar J)I_H+\beta_{M\leftarrow H} f'(\bar J) \bar I_H J-\mu I_M,\\[3mm]
	J(t) &=\displaystyle\int_0^{+\infty}I_H(t-\theta) K(\theta)\dd \theta.
\end{array} 
\right.
\end{equation}
To derive a characteristic equation, we look for solutions of the form 
\begin{equation}\label{expsol}
\left(I_H(t), I_M(t), J(t)\right)=\left(v_H, v_M, v_J \right){\e}^{\lambda t},\qquad  v_H, v_M, v_J,\lambda \in \C,
\end{equation}
with $(v_H, v_M, v_J)\ne (0,0,0)$ and $\Re(\lambda)>-\nu$. Substituting \eqref{expsol} into \eqref{lin:outbreak2}, we obtain
\begin{equation*}
\left\{\setlength\arraycolsep{0.1em}
\begin{array}{rl} 
	\lambda v_H  &= \left[\beta_{H\leftarrow M} \rho S_H f(\bar J)\right] v_M +\left[\beta_{H\leftarrow M}\bar I_M \rho S_Hf'(\bar J)\right] v_J - \gamma   v_H  ,\\[5mm]
	\lambda v_M &= \left[\beta_{M\leftarrow H}\rho f(\bar J)\right]v_H+\left[\beta_{M\leftarrow H}\rho f'(\bar J) \bar I_H\right] v_J-\mu v_M,\\[3mm]
	v_J &=\displaystyle\int_0^{+\infty}  {\e}^{-\lambda\theta} K(\theta)v_H\dd \theta,
\end{array} 
\right.
\end{equation*}
from which we derive the characteristic equation
\begin{equation}\label{char:eq}
\det\left(\Delta(\lambda)\right)=0   
\end{equation}
for the characteristic matrix
\begin{equation*}
\Delta(\lambda)=    
\begin{pmatrix}
	\lambda + \gamma   &-\beta_{H\leftarrow M}\rho S_Hf(\bar J)  & -\beta_{H\leftarrow M}\bar I_M \rho S_Hf'(\bar J)\\[3mm]
	-\beta_{M\leftarrow H} f(\bar J) & 
	\lambda + \mu & -\beta_{M\leftarrow H}\bar I_H f'(\bar J)\\[3mm]
	-\hat K(\lambda) & 0 & 1
\end{pmatrix},
\end{equation*}
where $\hat K$ is defined as in \eqref{laplace_erl}.
Then, the characteristic equation \eqref{char:eq} takes the form
\begin{align}
\notag\ (\lambda+\gamma)(\lambda+\mu)-& \beta_{H\leftarrow M}\beta_{M\leftarrow H}\rho S_H f^2(\bar J)\\\label{step_char_f} -&\ \hat K(\lambda)f'(\bar J)\left[\beta_{H\leftarrow M}\beta_{M\leftarrow H}\rho S_H\bar I_H f(\bar J)+(\lambda+\mu)\beta_{H\leftarrow M}\rho \bar I_M S_H\right]=0.
\end{align}
Using $\hat R_e$ as defined in \eqref{Re}, we can rewrite  \eqref{step_char_f} as
\begin{align*}
\lambda^2+\lambda(\gamma+\mu)+\gamma\mu\left[1 -\hat R_e^2f^2(\bar J)\right] -\hat K(\lambda)f'(\bar J)\left[\gamma\mu\hat R_e^2 \bar I_Hf(\bar J)+(\lambda+\mu)\beta_{H\leftarrow M}\rho\bar I_MS_H \right]=0.
\end{align*}
Moreover, from \eqref{equilibria} we have
$\bar I_M = \cfrac{\beta_{M\leftarrow H}}{\mu} f(\bar J)\bar I_H.$
Hence, we finally obtain
\begin{align*}
\lambda^2+\lambda(\gamma+\mu)+\gamma\mu\left[1 -\hat R_e^2f^2(\bar J)\right]- \hat K(\lambda)f(\bar J)f'(\bar J)\gamma \hat R_e^2(\lambda + 2\mu) \bar I_H=0,\qquad \Re(\lambda)>-\nu.
\end{align*}
Note that, since the ODE model \eqref{outbreak} does not depend on $\nu$, the condition $\Re(\lambda)>-\nu$ can be replaced by $\Re(\lambda)>-k$ in the case of \eqref{outbreak}.

\section{Computations for the stability analysis of the EE}\label{comp_EE_exp}
In this section, we provide further details on the computations involved in the linear stability analysis of the EE in \Cref{sec:stab_EE}.

\subsection{Exponentially fading memory}\label{subsec:comp_EE_exp}
Assuming $n=1$ in \eqref{erlang}, we are led to study the solutions of the equation \eqref{char_exp}.
Expanding the terms, we obtain the cubic equation
$\lambda^3+A\lambda^2+B\lambda+C=0$, where
\begin{align*}
A&\coloneqq \gamma+\mu+k,\\
B&\coloneqq k(\gamma+\mu+\gamma\delta),\\
C&\coloneqq 2\mu k\gamma\delta.
\end{align*}
Note that $A>0$ and, from \eqref{der:h}, we have $B,\ C>0$. Hence, by the Routh--Hurwitz criterion, the EE is locally asymptotically stable if and only if $AB-C>0$.
This condition reads 
$$k(\gamma+\mu+k)(\gamma+\mu)+k(\gamma+\mu+k)\gamma \delta-2\mu k\delta\gamma >0,$$
which, since $k>0$, is equivalent to 
$(\gamma+\mu+k)(\gamma+\mu)+\gamma \delta(\gamma+k-\mu)>0$.
This inequality can be rewritten in terms of $k$ as in \eqref{exp:condLAS}.
Hence, in the case of an exponentially distributed kernel, the Routh--Hurwitz criterion shows that only three scenarios are possible, as summarised in \Cref{propexp}.

\subsection{Erlang-2 distributed memory}\label{comp_EE_Erlang2}
Assume that $n=2$ in \eqref{char_exp}.
Then, the characteristic equation takes the form \eqref{char:erl}. 
which can be rewritten as 
\begin{equation*}
\lambda^4+A\lambda^3+B\lambda^2+C\lambda+D=0
\end{equation*}
where
\begin{align*}
A&\coloneqq \gamma+\mu+2k,\\
B&\coloneqq k(2\gamma+2\mu+k),\\
C&\coloneqq k^2\left(\gamma+\mu+\gamma\delta\right),\\
D&\coloneqq 2\gamma\mu\delta k^2  .
\end{align*}
Note that $A,\ B>0$ and, by \eqref{der:h}, we have $C,\ D>0$. To apply the Routh--Hurwitz criterion to a fourth-order polynomial, it is necessary to verify the conditions
\begin{equation}\label{1:hurwitz}
BC-AD>0
\end{equation}
and
\begin{equation}\label{2:hurwitz}
ABC-A^2D-C^2>0.    
\end{equation}
Since $A>0$, condition in \eqref{2:hurwitz} is equivalent to
\begin{equation}\label{3:hurwitz}
BC-AD>\cfrac{C^2}{A},  
\end{equation}
which in particular implies \eqref{1:hurwitz}.
Hence, it suffices to verify \eqref{3:hurwitz}.

We compute
\begin{align*}
ABC&=k^3(\gamma+\mu+2k)(2\gamma+2\mu+k)(\gamma+\mu-\gamma\delta)\\[2mm]
&=2(\gamma+\mu-\gamma\delta)k^5+5(\gamma+\mu)(\gamma+\mu-\gamma\delta)k^4+2(\gamma+\mu)^2(\gamma+\mu-\gamma\delta)k^3, 
\end{align*}
while $C^2=(\gamma+\mu-\gamma\delta)^2k^4$ and $A^2D=-8\gamma\mu\delta k^4-8\gamma\mu\delta(\gamma+\mu) k^3-2\gamma\mu\delta(\gamma+\mu)^2 k^2$.
Therefore, condition \eqref{3:hurwitz} can be rewritten as
$$k^2(\tilde A k^3+\tilde Bk^2+\tilde C k+\tilde D)>0,$$
where $\tilde A,\ \tilde B,\ \tilde C,\ \tilde D$ are defined as in \Cref{properlang}.
From \eqref{der:h}, it follows that $\tilde A>0$ and $\tilde D<0$, while the signs of $\tilde  B$ and $\tilde  C$ depend on the choice of the model parameters.
Since $k^2>0$ for $k\in \R\setminus\{0\}$, the problem reduces to determining the values $k$ such as $p(k)>0$, where $p$ is defined as in \Cref{properlang}.
Observe that \begin{equation*}
p(0)=D<0\quad\text{and}\quad \lim_{k\to+\infty}p(k)=+\infty.
\end{equation*}
Hence, there exists $k_+,\ k_->0$ such that $p(k)>0$ for $k>k_+$ and $p(k)<0$ for $k\in (k_-, k_+)$. The uniqueness of $k_+$ in $(0, +\infty)$ depends on the sign of $\tilde  B,\ \tilde C$. Applying the Routh--Hurwitz criterion then yields the result stated in \Cref{properlang}.
\end{appendices}

\bibliography{myreferences}

\end{document}